\documentclass[english,published]{programming}

\usepackage[backend=biber]{biblatex}
\usepackage{booktabs}
\usepackage{multicol}
\usepackage{ccicons}
\usepackage{todonotes}
\usepackage{amsthm}
\theoremstyle{remark}
\newtheorem{prop}{Property}

\newtheorem{theorem}{Theorem}[section]
\lstdefinelanguage[programming]{TeX}[AlLaTeX]{TeX}{%
  deletetexcs={title,author,bibliography},%
  deletekeywords={tabular},
  morekeywords={abstract},%
  moretexcs={chapter},%
  moretexcs=[2]{title,author,subtitle,keywords,maketitle,titlerunning,authorinfo,affiliation,authorrunning,paperdetails,acks,email},
  moretexcs=[3]{addbibresource,printbibliography,bibliography},%
}%
\newcommand*{\CTAN}[1]{\href{http://ctan.org/tex-archive/#1}{\nolinkurl{CTAN:#1}}}

\newcommand{\tlapluspdfpage}[2]{
  \begin{center}
  \includegraphics[width=1.0\textwidth,page=#2,trim=4.5cm 4cm 4cm 4cm, clip=true]{#1}
  \end{center}
}

\newcommand{\tlaplusnospace}{TLA$^{+}$}
\newcommand{\tlaplus}{TLA$^{+}$ }

\usepackage[normalem]{ulem}

\paperdetails{
  perspective=engineering,
  area={General Programming, Program Verification, Modularity and separation of concerns},
  license=cc-by-nc-nd
}

\begin{document}

\title{Designing a Producer-driven Stream Protocol by Formal Refinement}

\author[a]{Erick Lavoie}[0000-0002-4020-6578]
\authorinfo{is a post-doctoral researcher in the Computer Networks Group. Contact him at
  \email{erick.lavoie@unibas.ch}.}
\affiliation[a]{Department of Mathematics and Informatics, University of Basel, Switzerland}




\maketitle

\begin{abstract}
The coroutine, defined by Conway in 1963, has broadly diffused in the practice of concurrent programming. In the Python programming language, it appears both as generators and asynchronous functions. In the Unix operating system, it appears as processes communicating through pipes. We wanted to use coroutines in Python to create single-threaded Unix-style pipelines. Unfortunately, available solutions in Python are cumbersome to use: communication with a module depends on whether it is synchronous or asynchronous, does not provide flow-control, and must implement termination with exceptions.

The push-stream protocol, a producer-driven object protocol incepted by Dominic Tarr to obtain Unix-style pipelines in JavaScript, appeared to be a good alternative. However, its specification is incomplete and ambiguous: it focuses on the static structure of objects, is unclear regarding the termination status of intermediate modules, and at times inconsistent with implementations.

What behaviours should a push-stream protocol allow? How should the protocol be specified? How can module specifications be verified for conformity? To answer these questions, we have used \tlaplus and the TLC model checker to rederive the protocol and obtain protocol-specific verification tools.

In this paper, we present a formal specification of a push-stream protocol that 1) seamlessly combines synchronous and asynchronous modules, encapsulating the choice within each module; 2) provides flow control without using bounded buffers; 3) gracefully and unambiguously terminates; 4) does not require dynamic allocation of objects on the heap. In addition to completely describing expected behaviours, our specification improves on the original design by 1) allowing the input and output of intermediate pipeline modules to terminate independently and 2) explicitly reporting when a module is pending on the execution environment, to avoid incorrect resuming. 

We specify the protocol as a sequence of refinement steps: we first define a port specification abstracting object interactions; we then refine the abstract port into a specific input-output protocol based on method calls and boolean status flags; and we finally restrict possible cross-port invocations. From this protocol, we derive by equivalence a specification of what an abstract module may do. We then refine the latter into a module checker that can verify concrete module specifications for conformity. 

We have verified the key properties of all specifications and the validity of refinement steps with TLC. In supplemental material, we provide all \tlaplus specifications, show that the protocol is sufficiently expressive to implement a superset of all original JavaScript modules, as well as a performance comparison with Python alternatives and Unix pipes.

Our work offers a new approach to Python concurrent programming. It further serves as a reference to implement push-stream modules and module checkers in other object-oriented programming languages. It finally provides the first experience report on the design of object protocols and the derivation of protocol-specific verifiers by refinement using \tlaplusnospace.


\end{abstract}

\section{Introduction}
\label{sec:introduction}

Melvin E. Conway originally introduced the concept of a coroutine in 1963, in the design of a single-pass high-performance COBOL compiler~\cite{conway1963coroutine} and defined it as "a module that is coded as an autonomous program which communicates with adjacent modules as if they were input or output subroutines". The concept was later made more precise and general in different models of computation~\cite{kahn1976coroutines,hoare1978csp,lee1995dataflowprocessnetwork} and usable in many programming languages and systems~\cite{marlin1980coroutines,ritchie1984unix-evolution,schemenauer2001python-generators,demoura2004lua-coroutines,van2004concepts,moura2009full-asymmetric-coroutines,selivanov2015python-coroutine,boduch2015javascript,elizarov2021coroutine-kotlin,chen2023structured-concurrency}.

Over a period of 15 years, the Python programming language has progressively added support for Conway's coroutines. First by introducing a simple consumer-driven streaming protocol with iterators~\cite{yee2001python-iterator}, i.e. stateful objects that incrementally produce values. Second by making iterators easier to define by introducing the "yield" statement, which turns a function into a generator -- i.e. a function that returns an iterator~\cite{schemenauer2001python-generators}.~\footnote{The expressivity gain comes from the iterator returned by the generator implicitly maintaining the state of the function executed at the last suspension point defined by the yield statement.} Third by extending generators to also support injection of values by a producer, turning "yield" into an expression~\cite{vanrossum2005python-enhanced-generators}. Fourth, by adding asynchronous support for non-generator functions using the "async/await" syntax~\cite{selivanov2015python-coroutine}. And finally, by extending generators to support the same asynchronous operations~\cite{selivanov2016python-async-generators}.

These constructions are still somewhat cumbersome when trying to define composable modules. First, the caller has to handle different protocols depending on whether the target is synchronous or asynchronous. Second, the producer-driven version of the previous protocols does not support flow control, i.e. pausing and resuming. Third, termination and error handling require raising exceptions and implementing exception handlers at the calling sites. This became quite acute when we tried to implement processing pipelines that would be as easy to define as chaining a sequence  of Unix commands with pipes~\cite{ritchie1974unix,ritchie1984unix-evolution}. We needed a protocol that could hide the complexity within the callee and better handle flow-control and termination.

The push-stream protocol~\cite{pushstream}, a producer-driven object protocol designed by Dominic Tarr for JavaScript almost a decade ago, appeared to address our issues with Python. Serendipitously, the Python asyncio lowest-level stream protocol~\cite[Stream Protocols]{asyncio-rebooted} has an equivalent interface. However, our attempt to port push-stream to Python revealed ambiguities and holes. The original specification, in the form of a readme file~\cite{pushstream}, mostly focuses on the static structure of objects -- method signatures and public member variables -- and little of the dynamic behaviour. It is also unclear regarding how an intermediate module in the pipeline should handle termination and is at times inconsistent with the implementation~\footnote{"When writing to a stream, check the value of paused both before and after."~\cite[readme]{pushstream}. However, all implemented modules in the same repository never verify the paused status before writing. During our specification effort (Section~\ref{sec:port-state-machine}), we realized that in a linear pipeline, a module is never reactivated \textit{unless} the following module is ready to receive the value, so the implementation is correct and the specification overly cautious.}.

We completed the push-stream specification. This turned out harder than initially thought, because this protocol distributes and interleaves the responsibility of sequentially activating concurrent modules with that of ensuring that values are transferred between two adjacent modules only when both are ready. These two responsibilities are usually separately handled by a central scheduler and a communication port. We feel the effort was worth it because: 1) this solved our issues with Python and 2) the resulting protocol should be portable to other languages that support the equivalent of structs and function pointers, such as C~\cite{kernighan1988c} or Forth~\cite{rather1996evolution-forth}. To ensure correctness, we leveraged the \tlaplus formal specification language~\cite{lamport1994tlaplus} and the associated TLC model checker~\cite{lamport2002specifying-sytems}. This led to four major contributions, detailed hereafter.


First, our specification improves on the original protocol by: 1) making all possible executions explicit; 2) allowing the input of a transformer module to terminate independently of its output; 3) making explicit the pending state of an output during asynchronous value transfers, preventing confusion with the idle state in which it is waiting to be resumed. 
 
Second, we present the specification of the protocol as a sequence of refinement steps, which makes it easier to understand and verify than if it had been presented as a single artifact. The use of refinement in formal specifications is uncommon and has previously mostly been applied to consensus algorithms~\cite{tlaplus-community2016tlaplus-examples}. This paper provides a first example of using refinement to define an object protocol.
 
Third, we show how mutual refinement can be used to obtain different but equivalent specifications for conflicting practical purposes: one to describe the protocol concisely and the other to minimize implementation errors. We derive abstract modules by equivalence with the previous protocol description. The abstract modules simplify the task of specifying concrete modules because their style is close to that of concrete implementations in imperative object-oriented languages. 

Fourth, we refine abstract modules into a module checker. The module checker, when used in combination with the TLC model checker from the  \tlaplus toolbox, automatically identifies protocol violations of concrete module specifications.\footnote{The current version of the module checker may only verify concrete module specifications written in \tlaplusnospace. It should be possible and hopefully not too hard to port to run natively, thus removing the dependency on the \tlaplus toolbox.} This should make the implementation of correct modules easier.

The rest of this paper is organized as follows. We first illustrate the protocol with concrete module examples and provide a high-level introduction to \tlaplus (Sec~\ref{sec:background}). We then provide an overview of the push-stream protocol and the requirements the full specification meets (Sec~\ref{sec:protocol-overview}). We follow with the full specification of the protocol, by refinement (Sec.~\ref{sec:push-stream-spec}). We then evaluate the expressivity of the protocol and verification costs (Sec.~\ref{sec:evaluation}). We continue by relating this effort to previous work (Sec.~\ref{sec:related-work}). We then conclude with a summary and future perspective (Sec.~\ref{sec:conclusion}). Appendices provide all \tlaplus specifications discussed in the paper and a performance comparison to Python alternatives.

\section{Background}
\label{sec:background}

In this section, we illustrate the underlying concepts of the push-stream protocol with concrete module examples, followed by an overview of \tlaplus.

\subsection{Push-Stream Module Examples in Python}

Suppose we are designing an application in which end users may create streaming pipelines at run time, by connecting objects. At first, consider a \textit{producer} that sequentially outputs values $1$ to $n$, \texttt{PSeq}, and a \textit{consumer} that counts the total number of values received, \texttt{CCount}. We make the producer active and the consumer passive,\footnote{One alternative is to make the producer passive and the consumer active, with the consumer \textit{pulling} (reading) values from the producer. Another alternative is to make both actives but  allow a single value to be transferred only when both have previously \textit{pushed}-\textit{pulled}~\cite{hoare1978csp}.} i.e. the producer \textit{push} (\texttt{write}) values onto the consumer, and we implement this behaviour with a method call. In addition, we start the flow of values on the producer with a \texttt{resume} method. An implementation may look like that of Figure~\ref{fig:protocol-v1}.

\begin{figure}[h]
\begin{subfigure}[t]{0.45\textwidth}
\begin{lstlisting}[language=Python]
class PSeq:
    def __init__(self,n):
        self._i = 1
        self._n = int(n)
        self.sink = None
    def resume(self):
        while self._i <= self._n:
            self.sink.write(self._i)
            self._i += 1
\end{lstlisting}
\end{subfigure}
\hfill
\begin{subfigure}[t]{0.45\textwidth}
\begin{lstlisting}[language=Python,firstnumber=10]
class CCount:
    def __init__(self):
        self.n = 0
        self.source = None
    def write(self, v):
        self.n += 1

p = PSeq(3)
c = CCount()
p.sink = c; c.source = p; p.resume()
\end{lstlisting}
\end{subfigure}
\caption{\label{fig:protocol-v1} Push-stream protocol (v1) examples supporting \textit{writing} and \textit{resuming}.}
\end{figure}

Now suppose that rather than simply counting the values, we would like them to be appended at the end of a log file, in a \texttt{CLogFile} module. The associated file should be opened when the object is created and closed when all values have been written. We therefore extend the protocol with a new \texttt{close} method on the consumer and create a second version of the producer, \texttt{PSeq2}, in which \texttt{resume} closes the consumer when no more values will be written, as in Figure~\ref{fig:protocol-v2}.

\begin{figure}[h]
\begin{subfigure}[t]{0.45\textwidth}
\begin{lstlisting}[language=Python]
class PSeq2:
    ...
    def resume(self):
        while self._i <= self._n:
            self.sink.write(self._i)
            self._i += 1
        self.sink.close()
\end{lstlisting}
\end{subfigure}
\hfill
\begin{subfigure}[t]{0.45\textwidth}
\begin{lstlisting}[language=Python,firstnumber=8]
class CLogFile:
    def __init__(self, path):
        self._f = open(path, 'a')
        self.source = None
    def close(self):
        self._f.close()
    def write(self, v):
        self._f.write(v)
\end{lstlisting}
\end{subfigure}
\caption{\label{fig:protocol-v2} Push-stream protocol (v2)  extended with explicit \textit{closing}.}
\end{figure}

So far so good. What if we now want a \textit{transformer} module, say \texttt{TTake}, that writes to \texttt{CLogFile} up to a prefix of $m \geq 1$ values of the sequence generated by \texttt{PSeq2} and ignores the others? This module should appear as a consumer to \texttt{PSeq2} and a producer to \texttt{CLogFile}. As a first approximation, we could implement \texttt{TTake}'s \texttt{write} method such that it ignores any values after $m$; this works but is rather inefficient. We would instead want to \textit{abort} the producer once no more values are needed. A third version of the producer, \texttt{PSeq3}, uses an \texttt{abort} method to set an \texttt{ended} flag, and uses the flag to detect when it has been aborted while writing. We list an implementation of \texttt{PSeq3} and \texttt{TTake} in Figure~\ref{fig:protocol-v3}.
\begin{figure}[h]
\begin{subfigure}[t]{0.45\textwidth}
\begin{lstlisting}[language=Python]
class PSeq3:
    def __init__(self,n):
        self._i = 1
        self._n = n
        self.sink = None
        self.ended = False
    def abort(self):
        self.ended = True    
    def resume(self):
        while self._i <= self._n:
            self.sink.write(self._i)
            self._i += 1
            if self.ended:
                break
        self.sink.close()
\end{lstlisting}
\end{subfigure}
\hfill
\begin{subfigure}[t]{0.45\textwidth}
\begin{lstlisting}[language=Python,firstnumber=16]
class TTake:
    def __init__(self, m):
        self._m = int(m) # Assume m >= 1
        self.source = None 
        self.sink = None
    def abort(self):
         self.source.abort()
    def close(self):
        self.sink.close()
    def resume(self):
        self.source.resume()
    def write(self, v):
        if self._m > 0:
            self._m -= 1
            self.sink.write(v)
        if self._m == 0:
            self.source.abort()
\end{lstlisting}
\end{subfigure}
\caption{\label{fig:protocol-v3} Push-stream protocol (v3) extended with early \textit{aborting}.}
\end{figure}

In some cases, a module may not be able to immediately process values and decide to buffer values internally to process them later. The size of the buffer would be finite to bound memory usage. In order to not loose values, the module would need to \textit{pause} writing once the buffer is full and resume the source again once ready. We then extend our protocol with a \textit{paused} flag on a consumer, that is set after a write if the next should temporarily be interrupted. A bounded buffer and a fourth version of the producer that supports pausing and resuming is shown in Figure~\ref{fig:pausing-protocol}.

\begin{figure}[h]
\begin{subfigure}[t]{0.45\textwidth}
\begin{lstlisting}[language=Python]
class PSeq4:
    def __init__(self,n):
        self._i = 1
        self._n = n
        self.sink = None
        self.ended = False
    def abort(self):
        self.ended = True    
    def resume(self):
        while self._i <= self._n:
            self.sink.write(self._i)
            self._i += 1
            paused = self.sink.paused
            if self.ended or paused:
                break
        if self._i > self._n or self.ended:
            self.ended = True
            self.sink.close()
\end{lstlisting}
\end{subfigure}
\hfill
\begin{subfigure}[t]{0.45\textwidth}
\begin{lstlisting}[language=Python,firstnumber=19,morekeywords={upon}]
class TBoundedBuffer:
    def __init__(self, m):
        self._max = m
        self._q = deque()
        self.paused = False
        self.source = None
        self.sink = None
    def abort(self):
        self._q.clear()
        self.source.abort()
    def close(self):
        if len(self._q) == 0:
            self.sink.close()
    def resume(self):
        while len(self._q) > 0:
            self.sink.write(self._q.popleft())
            if self.sink.paused:
                return
        if not self.source.ended:
            self.source.resume()
        else:
            self.sink.close()
    def write(self, v):
        self._q.append(v)
        if len(self._q) == self._max:
            self.paused = True
\end{lstlisting}
\end{subfigure}
\caption{\label{fig:pausing-protocol} Push-stream Protocol (v4) extended with \textit{flow control}.}
\end{figure}


Now suppose that we would like a transformer \texttt{TSubprocess}, that wraps an operating system process's standard input and output.  The Python asyncio subprocess module provides an API with similar functionalities our protocol supports so far, except that the output of the process is produced \textit{asynchronously}, i.e. the pipeline's stack of method calls needs to unwind first to let the event loop process pending events and eventually trigger a callback to receive new values. 

One problem that arises, and is shared with the JavaScript push-stream protocol~\cite{pushstream}, is that it is now no longer possible to distinguish between a source that has been previously resumed and is \textit{pending} on asynchronous behaviour and another that has never been resumed. We therefore add a \textit{pending} flag on a producer (or transformer) that is true in the first case and false otherwise.\footnote{A module may also remember whether it has previously invoked \texttt{resume} on its source: that would be equivalent of putting the pending flag on the consumer instead.}

A second problem is that asynchrony increases the possibilities of inter-leavings of method calls with possibilities of some never triggering or triggering more than once. The reason is that the \texttt{resume} method may not always be active once termination is requested. For example, in \texttt{PSeq4}, an \texttt{abort} may happen after \texttt{resume} has returned because the sink was paused but not all values have been produced: this implies that \texttt{self.sink.close()} (l.18) will not be invoked. This may prevent closing of the rest of the pipeline. One alternative could be to close immediately the sink during an abort, and remove the "\texttt{or self.ended}" condition on l.16. However, in the synchronous case where an abort happens while the last value is written ($i=n$), the sink would be redundantly closed again. What we would like instead is to guarantee that \texttt{abort} and \texttt{close} are always called exactly once: we therefore add \texttt{closed} and \texttt{ended} flags to respectively track whether the input and output have been closed already.

The protocol is now quite complex with four boolean flags (\texttt{closed}, \texttt{ended}, \texttt{pending}, and \texttt{paused}) and four methods (\texttt{abort}, \texttt{close}, \texttt{resume} and \texttt{write}) for every pair of connected modules. What states and inter-leavings of module activations and method calls should be valid and why? And how could we verify that a module correctly implements the protocol? The rest of this paper explains the key concepts and the overall approach we used to specify the protocol and derive a module checker, using state machines and properties defined on behaviours written in \tlaplus.

\subsection{\tlaplus}

\tlaplus is a formal language designed by Leslie Lamport and originally introduced in the 90s~\cite{lamport1994tlaplus} to help specifying systems, with a focus on distribution and concurrency. The core semantics of the language are based on long-established mathematical concepts -- including boolean and predicate logic, integer arithmetic, set theory, and mappings -- and use Zermelo-Fraenkel set theory~\cite{fraenkel1973foundations} as a common foundation. This core is extended to support operators from Pnueli's temporal logic~\cite{pnueli1977temporal} to write safety and liveness properties. In more common terms, \tlaplus is a language for specifying concurrent programs as abstract state machines and mathematically reason about the executions that are consistent with them. In our experience, compared to other formal languages we briefly tried, such as Agda~\cite{bove2009brief}, writing specifications in \tlaplus is conceptually much closer to programming in JavaScript or Python.  This can be explained partly because \tlaplus is dynamically typed and partly because many programming constructs in those languages are relatively easy to map to \tlaplus concepts.

Explaining the syntax and semantics of \tlaplus  is beyond the scope of this paper. A reading proficiency of all the specifications on which the main part of the paper is based (Appendix~\ref{appdx:abstract-port-spec} and after) can be achieved relatively quickly using the excellent video course by Lamport~\cite{lamport2021tlaplus-video-course}. A longer and more detailed discussion on the design principles of \tlaplus is available in a recent book~\cite{lamport2026science}. Reading proficiency of \tlaplus is however not necessary to understand the main part of this paper: our presentation focuses on the concepts and the overall refinement process we used for the design, which are illustrated with minimal references to the \tlaplus formulation\footnote{We took great care in ensuring the figures are consistent with the specifications they illustrate, but in case of discrepancies, the latter should prevail.}.

The properties of a \tlaplus specification, i.e. statements that are true of all possible executions, may be verified both using the TLC model checker~\cite{lamport2002specifying-sytems} and the TLAPS proof assistant~\cite{chaudhuri2010tlaps}. The former may exhaustively verify a (finite) set of behaviours, i.e. sequences of states representing possible executions, using a configuration that makes the specification executable and, if necessary, constrains the state space. The second converts \tlaplus formulas and proof steps to formats understood by backend verifiers -- including Isabelle, Zenon, and an SMT solver -- and reports whether the solver could show they are true. In practice, the effort of verifying a specification with the model checker is at least an order of magnitude lower than proving the same property with the proof assistant, so the former is always used as a first step. Moreover, the proof assistant currently has very limited support for temporal reasoning, so verifying temporal properties mandates the use of the model checker.  Third, our specifications are sufficiently abstract to admit only a finite state space, with minimal or no artificial restrictions, so verification with the model checker is essentially sufficient.

\section{Overview of the Push-Stream Protocol}
\label{sec:protocol-overview}

An object protocol is a set of conventions for how to structure modules as objects, as well as when and how they can interact with each other. This paper focuses on giving a high-level but detailed description of a push-stream protocol and the approach we took to specify it. A practical implementation may enforce conformance, e.g., during execution, type-checking, or testing but such mechanisms are orthogonal to our discussion and will be covered in future work. 

An overview of the protocol is illustrated in Figure~\ref{fig:port-protocol}. We limit our presentation to linear streaming pipelines: such pipelines have a single producer and a single consumer which by convention are respectively located at the beginning and the end of the pipeline, with an arbitrary number of intermediate transformers in-between. Values flow from the producer towards the consumer. By further convention, we say that an element of the pipeline is located upstream when it is closer to the producer, and another element is located downstream when it is closer to the consumer. More complex processing graphs may be created by implementing modules with multiple inputs or outputs, but these are out-of-scope and may require extending the protocol.

\begin{figure}[htbp]
\begin{center}
\includegraphics[width=1.0\textwidth,page=4,trim=0cm 8cm 0cm 0cm, clip=true]{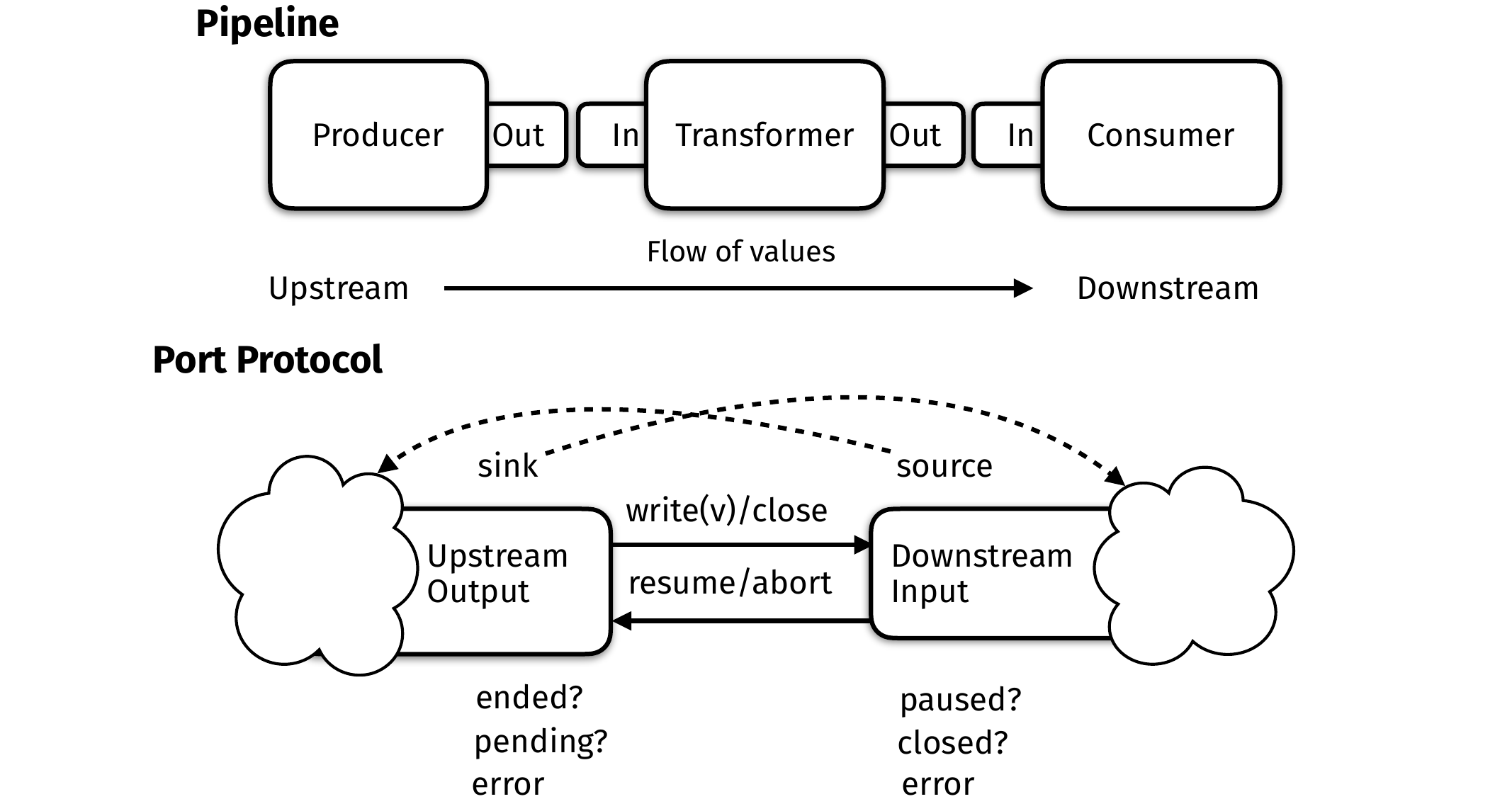}
\caption{\label{fig:port-protocol} Streaming Pipeline and Port Protocol Components}
\end{center}
\end{figure}

Modules have the following static structure. A module with an input port defines \texttt{write} and \texttt{close} methods, \texttt{paused} and \texttt{closed} boolean status flags, as well as a \texttt{source} reference to the upstream module (null if not connected). A module with an output port defines \texttt{resume} and \texttt{abort} methods, \texttt{ended} and \texttt{pending} boolean status flags, as well as a \texttt{sink} reference to the downstream module. All modules also define an \texttt{error} flag, possibly containing an error code or exception that arose during execution.

We present different aspects of the protocol in the next sections, by relating them to the static structure of modules. We also state the properties that the protocol provides for all possible executions.

\subsection{Modularity} 

The most apparent and useful aspect of the protocol is that it enables splitting complex processing pipelines in reusable and independent modules. To do so, it must guarantee:

\begin{prop}
 \label{rq:composability}  \textit{Composability}: Any two modules that correctly implement the protocol individually correctly implement the protocol when connected.
\end{prop}

This makes the verification of conformity of arbitrary pipelines linear in the number of possible modules rather than exponential in the number of module combinations. This also enables defining the behaviour of a complex module by combining elementary modules. To obtain composability, we enforce (see Theorem~\ref{thm:enc-implies-composability}):

\begin{prop}
 \label{rq:encapsulation}
  \textit{Encapsulation}:
Each module $m$ interacts only with adjacent modules, with each solely through the protocol port it shares.
\end{prop}

The main focus of our specification is on the interactions between adjacent modules with minimal requirements on their internal structure or behaviour. For example, we allow modules to have internal state, wrap input-output protocols from other libraries, call internal functions, and dynamically allocate memory for internal operations. However, we disallow modules to communicate through global variables or to inspect the state of non-adjacent modules.

Two adjacent modules are connected by an \textit{input-output port} (Figure~\ref{fig:port-protocol}): the output belongs to the module upstream and the input to the module downstream. Modules may be connected in any order. The entire pipeline forms a doubly-linked list and we say that a pipeline is complete when there are no unconnected inputs or outputs. 

The behaviour of a module is implemented in its port methods. The processing of a single value is done by the \texttt{write} method on the input, which accepts a single argument $v$, the value to be processed. A write is initiated by the output, we therefore say that the output \textit{pushes} values onto the input. The type of values is irrelevant and the values transferred are abstracted in all following specifications.\footnote{Concrete module specifications, mentioned but not presented in this paper, do explicit the values written.}

\subsection{Execution Model}
\label{sec:execution-model}

We say that a module is active when one of its methods or internal functions is currently executing. Modules activate in sequence, one at a time:
 
  \begin{prop}
 \label{rq:coop-activation} \textit{Cooperative Activation}: 
 A module is activated either by the execution environment if no module is currently active, or from an adjacent module, if one is active. An active module only becomes inactive when invoking another module's method, or returning from one of its methods. Invoking or returning from a method invocation deactivates the origin and activates the target.
 \end{prop}
 
Method execution can therefore never be interrupted, which simplifies reasoning.  However, a single incorrect module may deadlock the entire pipeline, which makes the verification of modules paramount.  This execution model is compatible with event loops as available in JavaScript or Python asyncio~\cite{asyncio-rebooted}. It can also be compatible with multi-threading in Python within modules when the Global Interpreter Lock (GIL) is active (default), because the GIL guarantees only one Python thread is executing at a time~\cite{python-gil}. This does not hold when deactivating the GIL while using multi-threaded code~\cite{gross2023optionalgil}, but use of the latter is still experimental at the time of writing.

\subsection{Synchronization}

Since different modules may wrap interfaces to different concurrent components of the execution environment, these components may not be ready at the same time. Taking inspiration from Communicating Sequential Processes~\cite{hoare1978csp}, we therefore enforce:

\begin{prop}
 \label{rq:synchronization}  \textit{Synchronization}: A value is transferred (written) from a module to an adjacent module downstream, only if both are ready.
\end{prop}

One case of synchronization is:
\begin{prop}
 \label{rq:flow-control}  \textit{Flow-Control}: While writing, an input may interrupt the flow of values upstream by signalling that it is no more ready.
\end{prop}

Flow-control allows a module to communicate internal limitations, e.g. the input buffer is full or an internal channel is not ready. An input signals upstream it is no more ready by setting \textit{paused} (to true) while writing. 

Another case of synchronization is when an output is pending on the execution environment to produce its next value. The output signals that it is not yet ready by setting \texttt{pending} while resuming.

Both cases are mutually exclusive because the output cannot stay pending while the input is paused because a value could be produced before the input is ready again. When the output is pending, a port will normally reactivate first on the output; similarly for the input when paused. A port side that deactivates while ready, e.g. because it called or returned from a method,  must stay ready until reactivation.


Synchronization precludes the need for a bounded buffer~\cite{dijkstra123coop-seq-procs} in ports. It is nonetheless still possible to interpose a bounded buffer as a module: this will affect performance but not correctness. When not using buffers however, flow-control signals must exhibit:

\begin{prop} 
 \label{rq:transitive-interruption}  \textit{Transitivity}: Non-ready signals, i.e. pending or paused, transitively propagate through all adjacent modules that do not or cannot internally buffer values.
\end{prop}

This happens downstream when pending and upstream when pausing.

\subsection{Termination}

During execution, the protocol should show:

\begin{prop}
 \label{rq:leniency} \textit{Leniency}: Any module may terminate at any time, normally or due to an error.
\end{prop}

To terminate, a module sets its termination flags: i.e. \texttt{closed} on its input and/or \texttt{ended} on its output. Termination then propagates with:

\begin{prop}
 \label{rq:gracefulness} \textit{Gracefulness}: A terminated module must eventually trigger termination on all adjacent modules.
\end{prop}

This is possible even when the adjacent output is pending or the adjacent input is paused because of leniency (Prop.~\ref{rq:leniency}). A module gracefully terminates by invoking \texttt{close} on an adjacent input and \texttt{abort} on an adjacent output after having set its own termination flags. A module does not need to immediately terminate on all its ports: e.g. it may receive a collection of values on its input port, be closed from upstream, then output all items individually before ending its output and closing downstream.

Termination is final:

\begin{prop}
 \label{rq:irreversibility} \textit{Irreversibility}: Once a module has terminated, it remains terminated.
\end{prop}

This greatly simplifies the design.  And optionally, to help debugging, the protocol may allow:

\begin{prop}
 \label{rq:auditability} \textit{Auditability}: All errors are recorded in the module they occurred and may be inspected after termination.
\end{prop}

Modules may record the cause of termination, or exceptions that arose during the termination process. Auditing does not affect the interactions between modules so we omit discussion of the \texttt{error} status in our specifications.

\subsection{Memory Requirements}

For efficiency, protocol interactions are:

\begin{prop}
 \label{rq:stack-based} \textit{Method-based}: Protocol interactions may allocate memory only on the stack in the form of method invocations.
\end{prop}

Modules may still dynamically allocate memory for internal operations.

\section{Specification by Refinement}
\label{sec:push-stream-spec}

While the previous properties may appear quite intuitive, there is significant complexity that arises because of: 1) the possible combinations of values of the \texttt{closed}, \texttt{ended}, \texttt{error}, \texttt{paused}, and \texttt{pending} flags on each port; 2)  the stack state since the \texttt{write}, \texttt{resume}, \texttt{close}, and \texttt{abort} methods may invoke each other across ports, 3) the possibility for modules to become activated either on a method invocation or method return, and 4) the possibility of pending modules to reactivate in arbitrary orders. Naively reasoning about the states of a processing pipeline is unwieldy. 

The main goal of specifying a protocol is to tame that complexity. The main design difficulty resides in making the protocol permissive enough to minimize restrictions for module implementers while restricting the possible behaviours sufficiently to make the behaviour of the entire pipeline easy to predict for all possible executions.

In the next sections, we develop the full specification of the push-stream protocol in three  refinement steps resulting in three progressively more concrete specifications each adding additional concepts and/or restrictions on the previous more abstract specifications. They are shown in the upper half of in Figure~\ref{fig:specification-overview}.
Specifying in three steps rather than one simplifies understanding. 

\begin{figure}[htbp]
\begin{center}
\includegraphics[width=1.0\textwidth,page=8,trim=0cm 0cm 0cm 0cm, clip=true]{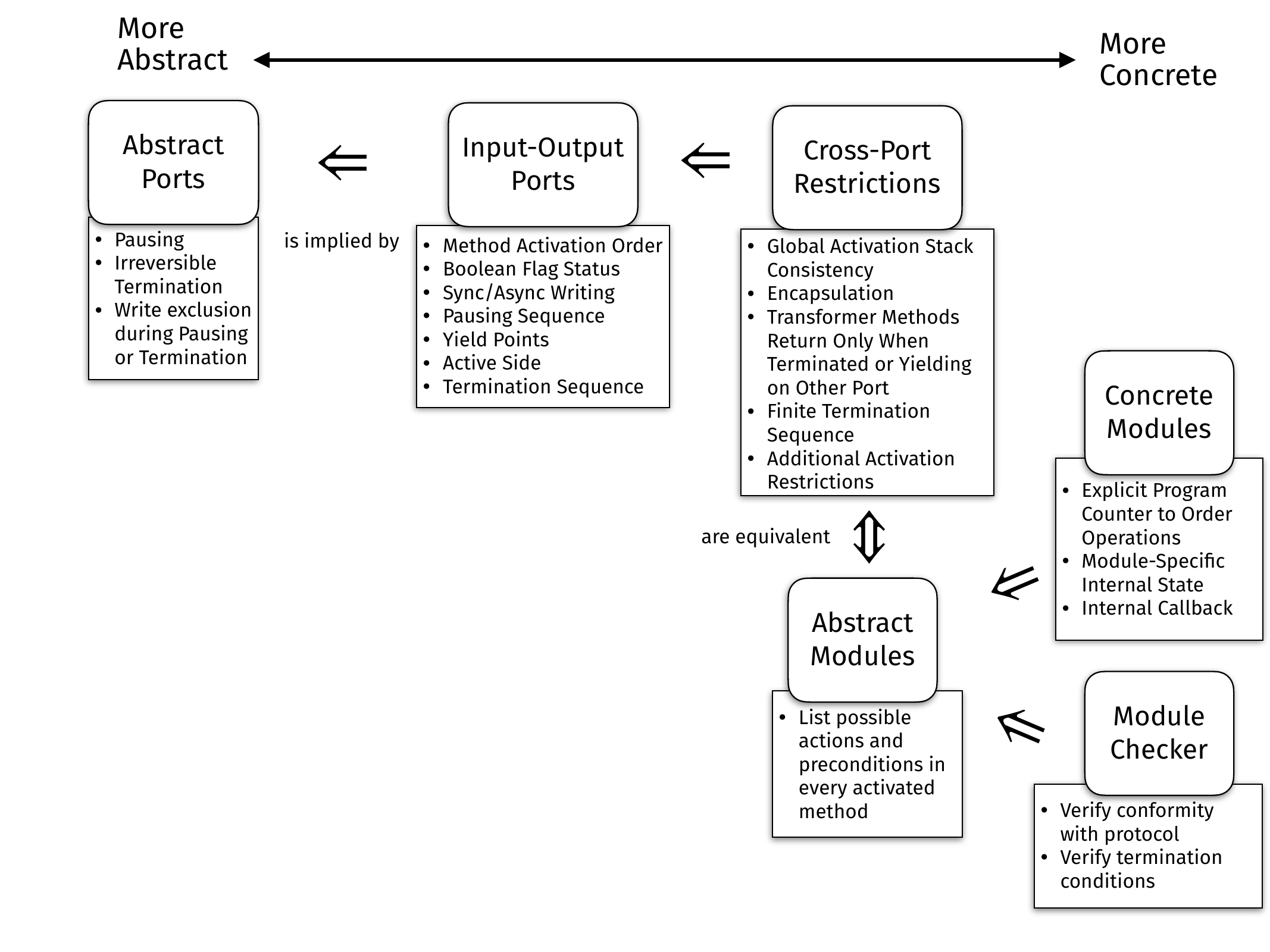}
\caption{Specification Overview. Refinement steps are shown with implication (i.e. $\Leftarrow$) and equivalence is shown with double implication (i.e. $\Leftrightarrow$). Key concepts introduced in each step are listed in square boxes.}
\label{fig:specification-overview}
\end{center}
\end{figure}

While Lamport and Merz orginally wrote that "[i]n practice we rarely care about checking equivalence of specifications"~\cite[p18]{lamport2017auxiliaryvariablestla} our experience suggests it may not be that rare. The lower half of Figure~\ref{fig:specification-overview} shows one instance of two equivalent specifications in which one version is easier and shorter to specify (i.e. "Cross-Port Restrictions") while the other is easier to refine into a practical implementation (i.e. "Abstract Modules"). This provides a concrete instance of equivalence that enables complementary specifications serving conflicting practical goals. 

The last two refinements at the bottom right of Figure~\ref{fig:specification-overview} respectively provide specifications of concrete modules to simplify later manual translation into imperative object-oriented languages, and a verifier to ensure their conformity with the protocol.


The development process did not follow the presentation order: we sometimes had to factor out abstract concepts from an overly complicated concrete specification; other times we had to restart from abstract principles to derive more concrete aspects. We also had to question original design choices from example JavaScript modules~\cite{pushstream}. The main lesson is that formal methods do not make the design process any less exploratory. They do help nonetheless to evaluate quickly and systematically candidate solutions and make the final result easier to verify.

\subsection{Abstract Port}
\label{sec:abstract-port}

The first and most abstract specification describes the behaviour of ports, which abstracts the specifics of what happens at the interface of two modules. The resulting state machine is illustrated in Figure~\ref{fig:port-abstract-state-machine} with the corresponding TLA+ specification listed in Appendix~\ref{appdx:abstract-port-spec}.

\begin{figure}[htbp]
\begin{center}
\includegraphics[width=0.5\textwidth,page=7,trim=10cm 4.5cm 8cm 4.5cm, clip=true]{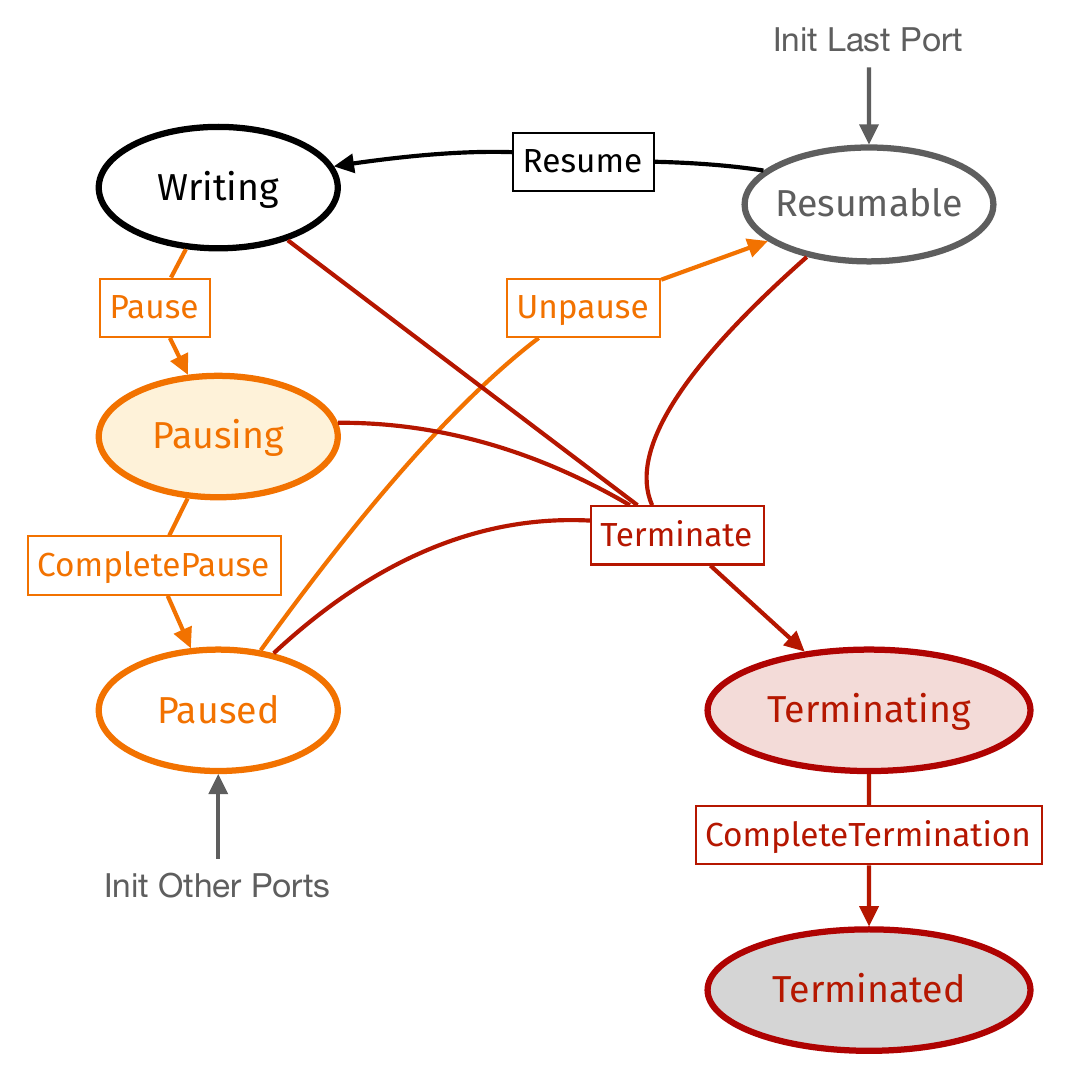}
\caption{\label{fig:port-abstract-state-machine} State machine describing the behaviour of a single port. The same colour scheme is used in Fig.~\ref{fig:port-state-machine} to illustrate refinement, i.e., how one state of the abstract port behaviour is implemented by one or more states of the input-output port behaviour.}
\label{default}
\end{center}
\end{figure}

A port is initialized either in the \textit{Resumable} state, if it is the last port of a pipeline, or in the \textit{Paused} state, for any other port. A resumable port may \textit{resume}, in which case it becomes \textit{Writing}. During writing, some values may be transferred between the corresponding output and input. At some point, the input may not be able to accept more values, in which case the port starts \textit{Pausing} until it is eventually \textit{Paused}. While pausing or paused, no writing should occur. Once the input is ready to accept more values, it unpauses the port, which becomes resumable again. From any of the \textit{Resumable}, \textit{Writing}, \textit{Pausing}, and \textit{Paused} states, the port may start \textit{Terminating} until it becomes \textit{Terminated}. Once terminating, no writing should occur again and once terminated, the port should stay terminated. This completes what a port is allowed to do, i.e. its safety property~\cite{lamport2026science}.

In addition, to guarantee progress, we assume all ports infinitely often take one of any possible action. To guarantee termination, we assume all ports eventually start terminating. Both are liveness properties~\cite{lamport2026science}.

\subsection{Abstracting the Stack State}
\label{sec:abstracting-stack-state}

Before presenting our next refinement, we discuss how the activation stack that is used to implement function and method calls in programming languages can be abstracted to focus on the relevant interactions between the corresponding inputs and outputs. 

One difficulty is that the implementation of modules may call functions or methods not part of the protocol specification. Our specification should not restrict nor mandate it. The next two specifications therefore ignore activation frames of functions and methods that do not belong to the protocol while preserving the order of the others.

An additional difficulty is that the methods of one port may be called as part of the implementation of the methods belonging to a different port, which in effect, results in the activation stack interleaving the activation of multiple ports. However, because of encapsulation (Prop.~\ref{rq:encapsulation}), we need a way to disentangle those activations to be able to reason about each port independently. 

\begin{figure}[htbp]
\begin{center}
\includegraphics[width=1.0\textwidth,page=5,trim=0cm 8cm 0cm 2cm, clip=true]{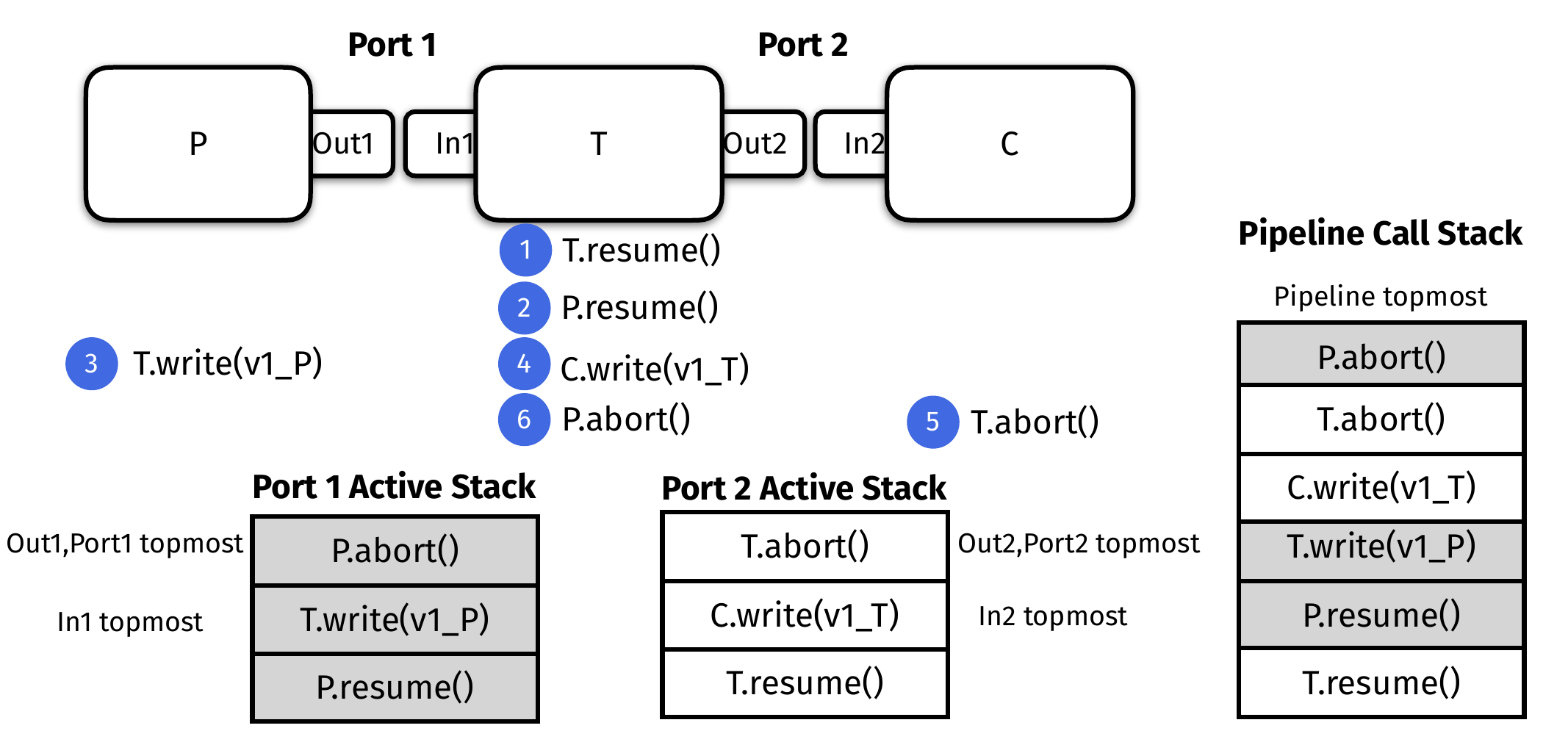}
\caption{Stack state of a simple Producer-Transformer-Consumer pipeline after a transitive resume with nested transitive write and further nested transitive abort.}
\label{fig:stack-state}
\end{center}
\end{figure}

Figure~\ref{fig:stack-state} illustrates one example showing a simple processing pipeline composed of a producer $P$, followed by a transformer $T$, and ending in a consumer $C$. This pipeline has therefore two input-output ports: one connecting the output of the producer to the input of the transformer and one connecting the output of the transformer to the input of the consumer. On the right, the activation stack for the entire pipeline is illustrated, after the following sequence of method invocations: 1) $T$'s output was resumed from the environment; 2) during $T$'s resuming, $T$ resumes $P$'s output; 3) during $P$'s resuming, it generates a first value \textit{v1\_P} that is written into $T$; 4) during $T$'s processing of \textit{v1\_P}, $T$ writes a modified value \textit{v1\_T} into $C$; 5) during $C$'s processing of  \textit{v1\_T}, $C$ decides that it requires no more values and aborts $T$'s output; 6) during $T$'s aborting, $T$ aborts $P$. The activation stack is shown before $P$ returns from aborting.

We disentangle the activation stack by considering, for each port, only the method activations of the corresponding input and output, while preserving the relative order of activation. This is illustrated at the bottom of Figure~\ref{fig:stack-state}. Conceptually, we can think of each port having their own activation stack executing concurrently with the others. The module whose method is currently at the top of the activation stack of the port, and the associated input or output, are considered active. A module may become active either after one of its method is invoked or after a method it invoked returns.  The input or output is considered active in a port even if the topmost method on the activation stack for the entire pipeline belongs to a different port. 

To simplify the presentation in the next refinement, we only mention the method name (and not the module on which it was called nor the arguments) when specifying the activation stack of a port. The rest can be deduced from the context because the input and output use different names for status variables and protocol methods. For example, the activation stack for Port 2 can therefore be summarized from top to bottom as \texttt{Abort}, \texttt{Write}, and \texttt{Resume}.

\subsection{Input-Output Port supporting both Synchronous and Asynchronous Writing}
\label{sec:port-state-machine}

We now describe the behaviour of a port's input and output. This refinement explicates a number of concepts: 1) the activation order of methods and the state of boolean flags; 2) which of the input or output (or none) is active in each state; 3) the expected sequence of steps required to pause or terminate, 4) how synchronous and asynchronous writing modes are supported, and 5) when a port may yield control to enable other ports or the execution environment to make progress. 

Possible states and transitions are illustrated in Figure~\ref{fig:port-state-machine} with a colour scheme that is consistent with Figure~\ref{fig:port-abstract-state-machine}, to highlight how a single state of the port specification abstracts a sequence of steps taken in the more concrete input-output specification. The corresponding TLA+ specification is listed in Appendix~\ref{appdx:io-port-spec}. We present the specification first by its initialization states and then according to the three main sequences of steps: synchronous writing and pausing, asynchronous writing and pausing, and termination.

\begin{figure}[htbp]
\begin{center}
\includegraphics[width=1.0\textwidth,page=6,trim=0cm 0cm 0cm 0cm, clip=true]{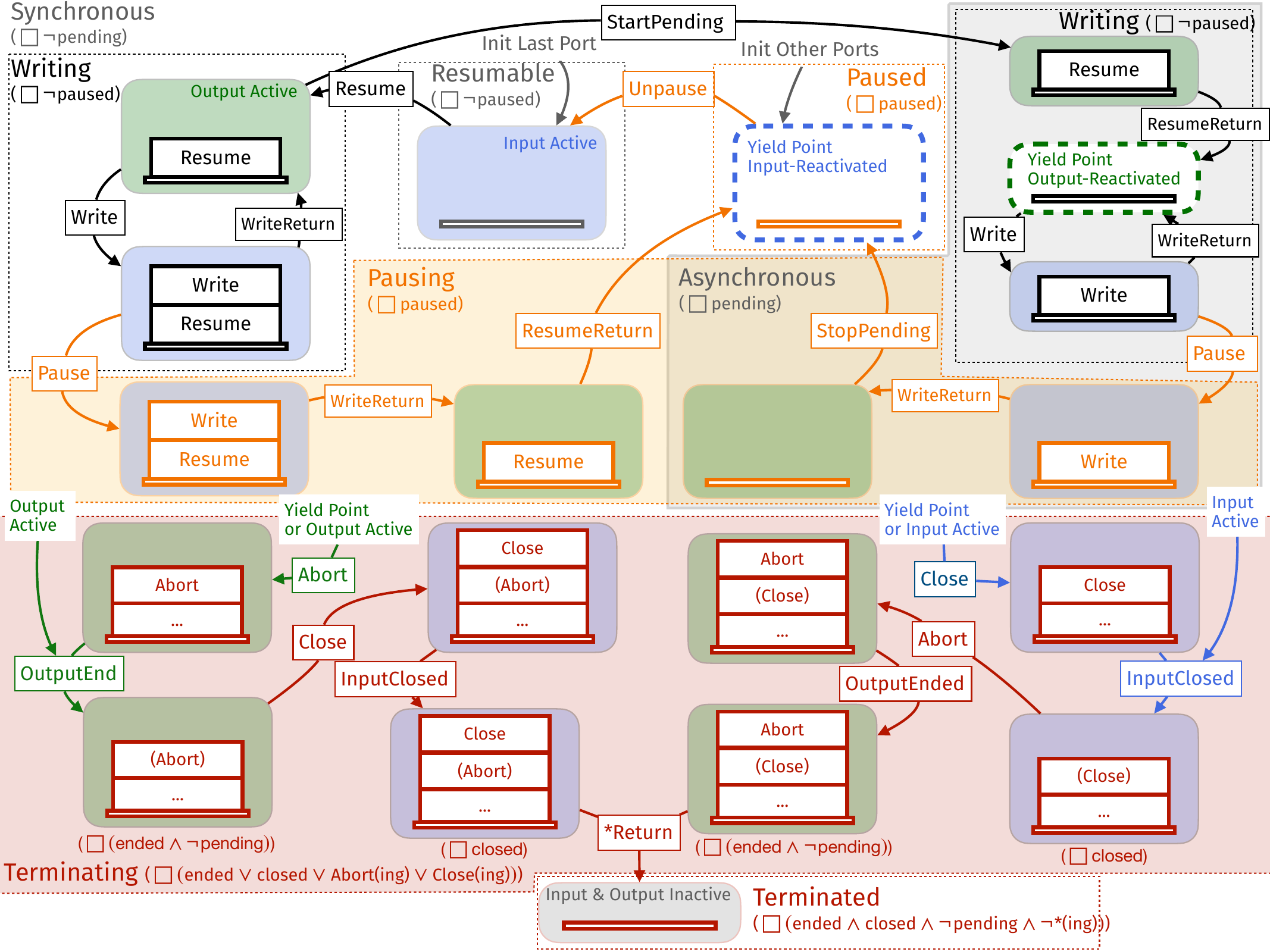}
\caption{\label{fig:port-state-machine} Port state machine showing possible input and output states and transitions. Each state is shown with a square box with rounded corner. The abstract stack state of the port is shown within each box with the top of the stack above, in a colour that matches the abstract state of the port specification. The background colour of each state indicates which of the input (blue), output (green) or none (gray) is active in that state, which determines the possible termination sequence that can be initiated from that state. Sets of states that share a similar status flag, e.g. pausing states always have paused set to true (written $\square \texttt{paused}$), are enclosed within dotted regions, which may overlap. All non-terminating states in the upper half of the figure negate the terminating conditions, i.e. \texttt{ended} and \texttt{closed} are never true and \texttt{Abort} and \texttt{Close} are never active.}
\label{default}
\end{center}
\end{figure} 

First, consistent with the abstract port specification, the last port is initialized in the \textit{Resumable} state: i.e. the input is not \texttt{paused} nor \texttt{closed}, the output is not \texttt{ended} nor \texttt{pending}, and none of the \texttt{abort}, \texttt{close}, \texttt{resume}, and \texttt{write} methods are active. For clarity, the \texttt{closed} and \texttt{ended} status flags are not listed for all non-terminating states (upper half of the figure): they are implied to be the opposite as when terminating, i.e. always both false. All other ports but the last are initialized \texttt{Paused}, that is identical to \texttt{Resumable} except that their input is \texttt{paused} and must unpause before resuming.

To resume synchronous writing, the input invokes the \texttt{resume} method on the output, which activates the output and places the resume method on top of the activation stack. While resuming, the output may invoke \texttt{write} on the input to transfer a value and activate the input to process it. After completing the processing of the value, the input may either return immediately -- ceding control back to the output and waiting for the next value to be written -- or after setting \texttt{paused} (to true) to interrupt further writing. In the first case, the output may immediately write another value, while in the second the output must cede control to the input by returning from the \texttt{resume} method and wait to be resumed again.

To resume asynchronous writing, the sequence is similar, except that the output sets \texttt{pending} and will return from \texttt{resume} before invoking \texttt{write}. This enables the output, e.g., to register a callback that will be invoked later to receive a message from an external channel. Returning from the \texttt{resume} method is necessary to cede control to the execution environment. The difference during asynchronous writing is that a \textit{write} is not invoked from the \texttt{resume} method but from the function that will be activated when the pending event fires. The input may distinguish between both cases by inspecting the \texttt{pending} flag on the output: i.e. if false the output is synchronously writing otherwise it is asynchronously writing.  If the input pauses during asynchronous writing, the output must stop pending on external events, e.g. by pausing the source of events or unregistering callbacks, before ceding control back to the input. This returns the port to the same \textit{Paused} state as for the synchronous case. The output may choose during execution whether to write synchronously or asynchronously and may decide differently for each value.

There are multiple states during which no protocol methods are active. They can be distinguished by the combination of the \texttt{paused} and \texttt{pending} values, where each combination unambiguously determines what the next step of the protocol should be. We discuss first all non-terminating behaviours. When \texttt{paused} and \texttt{pending} are false, the port is in the \textit{Resumable} state and must resume the output. When \texttt{paused} and \texttt{pending} are true, the port is pausing while the output is active so must stop pending and cede control to the input. We call the last two remaining states \textit{yield points} because these are the points at which the port cedes control to other ports upstream or downstream, or the environment. The two yield points differ by the expectation of which side the port must  reactivate on. When \texttt{paused} is true and \texttt{pending} is false, the port must reactivate from the input, unpause then resume the output. When \texttt{paused} is false and \texttt{pending} is true, the port must reactivate from the output and may start writing. If the output reactivates while paused, it should not write, while if the input reactivates when pending on a write, it should not resume the output.

We may now discuss the terminating sequence, which depends on which of the output or input is currently active.  First, when the output is active, it may invoke its own \texttt{abort} method, which should then set \texttt{ended}, or it may directly set \texttt{ended} without invoking \texttt{abort}. Before it sets \texttt{ended}, it should also unset \texttt{pending} (to false) and free up any resources that were needed for its operation. Once that is done, the output should trigger termination on the input by invoking the \texttt{close} method. During closing, the input should set \texttt{closed} and also free up any resources that were needed. After that is completed, all currently active methods should return, which results in the port being \textit{Terminated}. Alternatively, when the input is active, the sequence is similar, except that the termination order is reversed: i.e. the input closes before aborting the output. In both cases, the active part of the port should terminate first and set its termination flag, respectively \texttt{closed} or \texttt{ended}, before triggering the termination of the other part. This prevents infinite loops because the inactive part, when later activated, will not trigger termination back because the termination flag is already set. It is possible at the yield points, i.e. when either paused or pending on an asynchronous write, that the other side of the port activates first. In that case, it may still terminate the port which is necessary to ensure the protocol allows termination from any state (Prop~\ref{rq:leniency}). Termination is also always synchronous and should never pend on external events during the termination procedure.

\subsection{Restricting Interactions across Ports}
\label{sec:restricting-interactions-across-ports}

The previous input-output port specification allows all ports to execute concurrently and independently, effectively allowing any combination of port state across the entire pipeline. We now add restrictions on the valid combination of states across ports and calling context for port methods. This provides three main benefits: 1) it provides encapsulation (Prop~\ref{rq:encapsulation}); 2) it guarantees that modules terminate in a finite number of steps without requiring inspection of the stack state; 3) it removes redundant behaviours during termination to simplify the protocol. We present the concepts and properties that are enforced to achieve these goals. The corresponding formal safety properties are listed at the end of Appendix~\ref{appdx:cross-port-restrictions-spec}.

The first concept we add is a global activation stack, that we had incidentally previously abstracted in Section~\ref{sec:abstracting-stack-state}. It contains stack frames, each with the method name that was previously invoked -- but has not returned yet -- as well as the index of the module in the pipeline on which this method was called. This is necessary because the activation order between ports, determined by the invocation and return order of methods, is otherwise lost. We ensure that the global activation stack stays consistent with each port's active stack, by ensuring that 1) whenever a port activates a method, a method with the same name as well as the index of the module that corresponds to that method is added on the global activation stack; 2) whenever a method returns and is removed from a port active stack, the same method and corresponding module were on the global activation stack and are removed at the same time. The global activation stack allows the specification to add restrictions on the calling context, which we use in the next properties. Moreover, it is possible for the global activation stack to contain frames of functions or methods that are not part of the push stream protocol, i.e. \texttt{abort}, \texttt{close}, \texttt{resume}, \texttt{write}. For example, we use a special "activemodule" frame, that abstracts the activation of a module that comes from the execution environment.

Given the global activation stack, we enforce encapsulation (Prop.~\ref{rq:encapsulation}) with three conditions. First, methods may only be invoked when the module currently active -- whose method or internal callback function is currently on top of the global activation stack -- is either the module adjacent on the same port, or the same module. This prevents invocation of methods from non-adjacent modules and guarantees that during method invocation, only adjacent modules are activated. Second, status flags, i.e. \texttt{closed}, \texttt{ended}, \texttt{paused}, \texttt{pending}, may only be modified during the activation of a method of the corresponding module. This enforces that all other modules of the pipeline may only read their status, not modify it. Third, the \texttt{write} method should only be invoked by the module upstream, which follows from the assumption that modules are part of a linear pipeline in which values only flow from a producer upstream towards a consumer downstream (Sec.~\ref{sec:protocol-overview}).

The input-output protocol (Sec.~\ref{sec:port-state-machine}) introduced yield points and a final terminated state, enforcing that the port should deactivate only after the port either terminated, i.e. both the output is ended and the input is closed, or is yielding, i.e. the output is not ended and the input is not closed and either the output is pending or the input is paused. We now enforce this in three additional cases that were not considered previously. First, when a module returns from being activated by the execution environment, i.e. \texttt{activemodule} is removed from the global stack. Second, on the upstream port of a transformer when its output methods return, i.e. \texttt{resume} and \texttt{abort}. Third, and symmetrically, on the downstream port of a transformer when the transformer input methods return, i.e. \texttt{close} and \texttt{write}. This guarantees that deactivation of ports across the pipeline always leave them either terminated, in which case they never reactivate, or yielding, in which case they may reactivate on the input or output. A port is therefore never deactivated in the middle of termination.

The input-output protocol previously implicitly guaranteed termination with a finite number of steps by inspecting the activation stack of ports, and therefore preventing the same \texttt{abort} and \texttt{close} from being called more than once during termination. This is however inefficient and possibly impossible in some languages. To obtain the same property without requiring inspection of the stack, we instead restrict the caller to have terminated its side of the port(s) before terminating others. Specifically, if the \texttt{abort} method of a transformer is called on the output, the input of the same transformer must be closed. Otherwise if \texttt{abort} is called from a downstream transformer, then both the input and output of that transformer must have terminated before. Symmetrically, the same restrictions are applied in the downstream direction, i.e. when the \texttt{close} method of a transformer is called, and the restrictions applying in the opposite direction. Moreover, the same restriction also applies between the two ports of a transformer: the termination flag of one side must be set before setting the termination flag of the other side (without calling \texttt{abort} or \texttt{close}).

Finally, we enforce additional restrictions to simplify the protocol. First, \texttt{resume} and \texttt{write} should not be called during an \texttt{abort} or \texttt{close}. In synchronous mode, these are redundant because the target module may be reactivated during a method return. However, in asynchronous mode on a transformer during its input's close, we make an exception to allow pending values to be written downstream since the input will never reactivate. Second, we ensure that \texttt{write}s propagate across a transformer only when both ports are either synchronously writing or asynchronously writing, i.e. \texttt{pending} is true or false on both ports. This prevents a write from occurring on a pending transformer output port when the input has been synchronously written to, which should not arise in practice because the write would instead be triggered by an asynchronous event. Third, we also prevent unpausing or start pending during an \texttt{abort} or \texttt{close}. Fourth, we only allow pausing when a write method is active. The previous restrictions are not strictly necessary but they significantly decrease the number of cases required to support when implementing modules, while allowing the full range of interesting behaviours to still be implemented in practical modules.

This complete our presentation of the push-stream protocol, which was defined from the callee's perspective: i.e. the conditions under which the protocol methods may be invoked and returned from and the status flags changed. This had two main benefits: 1) it focused our attention on the interactions between the modules, rather than the modules themselves; 2) it made the specification of the protocol shorter because the methods may be called from multiple places with similar calling conditions. 

\subsection{Deriving Abstract Modules from the Protocol by Equivalence}
\label{sec:abstract-modules}

When implementing a module, the possible interactions need to be considered from a caller's perspective, i.e. from each of the \texttt{abort}, \texttt{close}, \texttt{resume}, and \texttt{write} methods. We consider the three cases of consumer, producer and transformer as abstract modules by only considering what is allowed to correctly implement the push-stream protocol. The result is provided in Appendix~\ref{appdx:abstract-modules-spec}.

This translation poses no special conceptual difficulty as it simply expresses the same conditions but with a different focus. While mathematically equivalent, the abstract module specification was easier to use to determine some of the cross-port restrictions, such a the conditions under which a write should be allowed (Sec.~\ref{sec:restricting-interactions-across-ports}).

 Nonetheless, some of our earliest attempts that attempted to model the abstract modules directly resulted in specifications thousands of lines long and painful debugging sessions. The two views were therefore complementary and we kept iterating on the design until specifications from both perspectives looked as simple as we could manage, both to explain the protocol and implement modules with.

\subsection{Specifying Concrete Modules}
\label{sec:concrete-module-spec}

Concrete modules, when implemented in common programming languages such as Python, execute method statements sequentially. This enables the control-flow within a method to implicitly determine pre-conditions on following statements. To replicate this behaviour, we specify the expected behaviour of concrete modules by 1) prefixing possible actions within a method with a unique index and 2) tracking the last statement that has been executed globally within each active method with a \textit{program-counter stack}, abbreviated \textit{pcstack}. This is standard practice in \tlaplus~\cite{lamport2026science}.

Using this style, the specification of concrete modules is quite close to the code one would write, e.g., in Python. We still abstract some details such as parameter passing, which includes the values being transferred between modules during a write. Previous attempts showed that these additional details added significant complexity while not contributing much to finding interesting issues.

\subsection{Verifying the Specification of Concrete Modules}
\label{sec:verifying-concrete-modules-spec}

We use the specification of abstract modules to verify that concrete modules behave according to the push stream protocol in all possible cases. To do so, we connect an abstract consumer or producer module to the port(s) of concrete modules and verify that the resulting pipeline behaves like a pipeline of abstract modules under a refinement mapping that abstracts all the details specific to the concrete modules, such as module-specific fields and stack frames. We provide the specification for this approach in Appendix~\ref{appdx:verifying-concrete-modules}. Since abstract modules and concrete modules only differ by the use of a \textit{pcstack} by the latter (Sec.~\ref{sec:concrete-module-spec}), the verifier ensures the pcstack stays consistent with the expected calling conventions when an abstract module invokes the method of a concrete module or returns from one of its methods.

When using model-checking to verify concrete modules, i.e. by exhaustively checking properties over a finite set of behaviours, the guarantees obtained depend on the set of behaviours verified and may not generalize. For example, a concrete module could correctly follow the push-stream protocol for up to the first 3 values written and later arbitrarily deviate from the protocol. Model-checking will not catch this issue unless at least four values are written. In practice however, this is not an issue because most useful modules either exhibit behaviour independent of the number of values written, or that varies according to user-parametrizable bounds, e.g. a module that counts from 1 to 3, while exhibiting the same boundary behaviour no matter the bound chosen, e.g. end when there are no more values to count. Our verification approach verifies all possible behaviours according to the initialization parameters for modules, so as long as a module's behaviour is solely determined by user-defined parameters and the behaviour of adjacent module(s), this verification approach will catch any deviations from the push-stream protocol and should provide sufficiently strong guarantees. Even if a module exhibits infinite behaviour under some parameters, e.g. count from 1 to infinity when the upper bound is null, we may still guarantee this infinite behaviour follows the push-stream protocol as long as the possible module actions are finite and the actions chosen during an infinite execution are a subset of those of the finite behaviour.

There still remains the possibility that a concrete module does not fulfill its expected behaviour, i.e. a module that is supposed to count from 1 to 3 but sometimes skips over 2, while still correctly following the push-stream protocol. Correctly specifying and verifying the behaviour of concrete modules beyond conformity to the push-stream protocol is outside the scope of this paper.

\subsection{Termination Conditions}

So far, except for a small discussion of the liveness assumptions on abstract ports (Sec.~\ref{sec:abstract-port}) our presentation has focused on safety properties, as determined by the valid sequence of states modules and the activation stack may go through. We now discuss the necessary conditions for a pipeline to terminate, which is the only liveness property we assert over the protocol. 

Note that the push-stream protocol does not mandate termination: a pipeline may be correct even if it never terminates, e.g. when using it to implement a long-running process that reacts to external messages before updating a shared local database. We now focus only on the internal conditions on pipeline modules that are necessary to ensure \textit{graceful termination} (Prop~\ref{rq:gracefulness}). 

The first condition is standard when modelling any program, sequential or concurrent: if a module may take a valid action then it must eventually do so~\cite{lamport2026science}. Formally, we assert that if any module is currently active and some of its actions remain enabled infinitely long, then it must eventually take one of these actions. This is ensured in practice simply by the semantics of the programming language in which the module is implemented: i.e. statements are executed one after the other until either the program terminates or an error occurs. Implementing a concrete module in any common programming language meets this condition automatically.

Second, we assume that if a module yields after internally pending on an event triggered by the execution environment, then the environment will eventually trigger that event and reactivate the module, allowing it to either unpause and resume upstream or write downstream. This condition is an assumption on the environment and therefore does not require anything from a module implementation per se.

Third, some modules may eventually  \textit{self-terminate}, i.e. terminate due to internal conditions -- e.g.  a counting module has counted up to the maximum number requested. Others may only terminate when explicitly requested to by another module, e.g. a map transformer that applies a function to every value it receives and otherwise waits for new values to arrive. Most of the concrete transformers we have implemented are examples of the latter category. A pipeline composed of modules only of the second category will never terminate. For a pipeline to terminate without external termination triggers, there must therefore be at least one self-terminating module.

Fourth, transformers may have two ports and the push-stream protocol allows each port to terminate independently of the other. To guarantee termination of the pipeline, a transformer must eventually start terminating on all of their ports once they start terminating on one. For most transformers we implemented, this is done as soon as the termination flags have been set on the first port terminating. In general, there might however be a delay: e.g. a flatten transformer may close its input and continue writing on its output until no more values are left to write (see Table~\ref{tb:concrete-modules}). Care must however be taken to ensure that this condition is satisfied regardless of which of the input or output terminates first, and even if termination is triggered on only one of the two.

Given that all previous four conditions are satisfied, a pipeline should always terminate. The first two conditions cannot be verified formally, since they are respectively assumptions respectively on the programming language and execution environment, but the latter two can. To verify the third condition, we use our model checker (Section~\ref{sec:verifying-concrete-modules-spec}) but assert no liveness properties that force the abstract producer or consumer to terminate. A verified module is \textit{self-terminating} if and only if the verification pipeline terminates in all cases. We verify the fourth condition by asserting that either the abstract producer or the abstract consumer must terminate (or both) but not necessarily in the same execution, and verifying that the verification pipeline including the transformer terminates in all cases. 


\section{Evaluation}
\label{sec:evaluation}

\subsection{Permissiveness}
\label{sec:expressivity}

A protocol should be sufficiently permissive to support useful modules. Moreover, the changes we have made to the original protocol should not preclude writing modules with equivalent behaviour. We have ensured both by specifying, verifying and implementing a superset of modules originally provided in JavaScript in the push-stream repository~\cite{pushstream}. We have specified the concrete modules with \tlaplus and verified them with our module checker. We have also implemented and tested each in Python. We summarize the list of modules in Table~\ref{tb:concrete-modules}.

\begin{table}[t]
\caption{\label{tb:concrete-modules} Summary of concrete modules specified in \tlaplus and implemented in Python. Module names with an asterix have no equivalent in the original push-stream repository.~\cite{pushstream}}
\begin{center}
\begin{tabular}{@{}lllr@{}} \toprule
Category   & Name & Sync./Async. & Description \\ \midrule
Producer   & Empty & Sync. & Immediately end on resume. \\ 
                  & Values & Sync. & Output iterable values. \\ 
Transformer & AsyncMap & Async. & Apply async. func. to every value. \\
                     & Batch* & Sync. & Group in lists of up to $n$ values each. \\
	            & Filter & Sync. & Output value if predicate is true. \\
                     & Flatten & Sync. & Output items of every iterable value. \\
                     & Map & Sync. & Apply func. to every value. \\
                     & Subprocess* & Async. & Connect in/out to subp. std in/out. \\
                     & Splitlines* & Sync. &  Output lines across input bytes/strings. \\
                     & Take & Sync. & Output first $n$ values. \\
Consumer  & Count* & Sync. & Count values. \\ 
                   & Collect & Sync. & Collect values. \\ 
                   & Drain & Sync. & Discard values. \\
                   & Reduce & Sync. & Reduce stream to a single value.\\
\bottomrule
\end{tabular}
\end{center}
\end{table}

\subsection{Verification with Model-Checking}
\label{sec:verification}

We verified the properties of all specifications and refinements on a Macbook Pro 2024, using an Apple M4 Pro processeur, 48 GB of RAM, running Tahoe 26.6.2, the OpenJDK (Java) 26.0.2.1, and the TLC model checker version 2026.09.17.032053 part of the 1.8.0 release of the \tlaplus command-line tool (\texttt{tla2tools.jar}).


\subsubsection{State Space and Time}

The state space size and time required for model-checking are listed in Table~\ref{tb:verification-cost-specifications}. Verifying all the specifications, refinements, and concrete modules takes less than ten minutes on a recent laptop, with the longest specification taking less than 2.5 minutes each, making the approach quite accessible.

\begin{table}[ht]
\caption{\label{tb:verification-cost-specifications} State space size and time required to verify the specifications and concrete modules.}
\begin{center}
\begin{tabular}{@{}llr@{}} \toprule
Name & Distinct States Found &  Time (s) \\ \midrule
Abstract Ports (Sec.~\ref{sec:abstract-port}, Appdx.~\ref{appdx:abstract-port-spec}) & 42  & $<1$  \\
IO Ports (Sec.~\ref{sec:port-state-machine}, Appdx.~\ref{appdx:io-port-spec}) & $11772$ & $1$ \\
Cross-Port Restr. (Sec.~\ref{sec:restricting-interactions-across-ports}, Appdx.~\ref{appdx:cross-port-restrictions-spec}) & $5504$ & $<1$ \\
Abstract Modules (Sec.~\ref{sec:abstract-modules}, Appdx.~\ref{appdx:abstract-modules-spec}) & $116307$ & $144$  \\
IO Ports $\Rightarrow$ Abstract Ports (Appdx.~\ref{appdx:refinement-io-port-to-abstract-port}) & $11772$ & $2$ \\
Cross-Port Restr. $\Rightarrow$  IO Ports (Appdx.~\ref{appdx:refinement-cross-port-restrictions-to-io-port}) & $91566$ & $10$ \\
Cross-Port Restr. $\Rightarrow$  Abstract Mod. (Appdx.~\ref{appdx:crossports-ref-abstract-modules}) & $91566$ & $18$ \\
Abstract Mod. $\Rightarrow$  Cross-Port Restr. (Appdx.~\ref{appdx:abstract-modules-ref-crossports}) & $116307$ & $141$ \\
Concrete Modules w/ Mod. Checker (Appdx.~\ref{appdx:verifying-concrete-modules}) & 13600 & 87 \\
\bottomrule
\end{tabular}
\end{center}
\end{table}

\section{Related Work}
\label{sec:related-work}

In this section, we present related work that served as inspiration and is also related in theme or techniques.

\subsection{Push- and pull-stream}

Over a decade ago, while designing and implementing Secure-Scuttlebutt~\cite{tarr2019ssb} in JavaScript, Dominic Tarr had issues with the Node.js stream API~\cite{tarr2016history-of-streams}. He addressed them by designing \textit{pull-stream}~\cite{tarr2016pullstream}, a consumer-driven function protocol based on a request-response format that imitates a callback API. This protocol, while quite successful in attracting a large number of community-contributed modules~\cite{pull-stream-modules}, required the allocation of a closure for every protocol interaction. This puts significant pressure on the garbage collector. Tarr later designed an alternative object-based producer-driven version called \textit{push-stream}~\cite{pushstream}.

The original push-stream protocol from Tarr~\cite{pushstream} deviates from our presentation in two major ways: 1) the input and output of a transformer share the same \texttt{ended} termination flag, which prevents distinguishing the termination of the input and output and may lead to writes being invoked after the output is marked as terminated; 2) the pending status of an output during asynchronous writing is kept implicit, which prevents distinguishing a pending write from the resumable (passive) state. While the protocol was used in practice, the detailed design process was also never documented and the protocol had not been formally verified.

The previous pull-stream protocol is similar to the push-stream protocol but significantly simpler to describe: i.e. the state machine required to describe the interaction between an input and output~\cite[Sec.~4.6.3 Fig.~4.26]{lavoie2020phdthesis} requires only four states and six transitions rather than the more than twenty states and transitions needed for our push-stream protocol (Sec~\ref{sec:port-state-machine}). The reason is that it basically treats every protocol interaction as asynchronous, through a callback-based request-response protocol, and does not expose the internal state of modules as status flags. However, it requires the allocation of one closure to handle the response of every interaction between modules.

\subsection{Reactive programming}
The push-stream protocol is related to reactive programming~\cite{bainomugisha2013survey,staltz2014intro-reactive-programming} in that it is also used to implement declarative processing pipelines. Academic work~\cite{bainomugisha2013survey} has emphasized a programming model based on reactive variables, whose values change in reaction to events from the execution environment, while popular programming framework directly expose and manipulate the underlying streams~\cite{staltz2014intro-reactive-programming}. This paper focused on specifying the lowest-level protocol that is used to efficiently implement modules, in a way that effectively removes the need for a scheduler by decentralizing its operations across the modules.

\subsection{Communicating Sequential Processes}
The push-stream input-output protocol (Sec.~\ref{sec:port-state-machine}) is related to the the input-output synchronization primitive from communicating sequential processes (CSP)~\cite{hoare1978csp}, albeit with explicit interleaving of deactivation/activation of adjacent modules and the input always triggering before the output. The equivalence is as follows. When the input unpauses, it marks itself ready to read a new value (equivalent to the CSP input command); when the input resumes the output, the input deactivates and transfers activation to the output (implicitly performed in CSP by the scheduler); when the output writes, it marks itself ready for writing a value (equivalent to a CSP output command); when the write method executes, the value is actually transferred between the output and the input, the output deactivates, and the input activates again (combining the CSP transfer and implicit activation of the input process by the scheduler).

\subsection{Meta-Stream Protocol}
De Troyer, Nicolay and De Meuter have proposed a meta-stream protocol~\cite{troyer2021metastreamprotocol}. Their work is orthogonal to ours: we have focused on verifying the correctness of one particular push protocol while their approach aims at providing facilities to allow users to tailor a streaming API to their needs, e.g. choosing push versus pull semantics and possibly implementing different value-propagation protocols. 

\subsection{\tlaplus and Refinement Mapping}
\tlaplus has been used to model and verify a variety of algorithms and protocols~\cite{tlaplus-community2016tlaplus-examples}. The published specifications most related to this work include a specification of a Byzantine Fault-Tolerant variant of Paxos, obtained by refinement~\cite{lamport2011byzantizingpaxos}, a specification of the TCP protocol~\cite{kuppe2026tcp-spec} as described by RFC 9293~\cite{ietf2022rfc9293}, and a specification to reproduce previous bugs in the Python asyncio lock library~\cite{alexn2021asynciobug}. The vast majority of published TLA+ specifications~\cite{tlaplus-community2016tlaplus-examples} do not need (or use) refinement and published examples are essentially related to consensus algorithms. To our knowledge we are the first to 1) design an object protocol in TLA+; 2) use refinement to verify the correspondence between the protocol description and a module checker that can verify conformity of concrete module specifications with the protocol description.

The existence of refinement mappings for any specification that can be written in TLA+ was first shown by Abadi and Lamport~\cite{abadi1991refinement-mapping}. A more accessible and practical presentation is given in a recent book by Lamport~\cite[Chapter 6]{lamport2026science} and an older paper by Lamport and Merz~\cite{lamport2017auxiliaryvariablestla}.

\section{Conclusion}
\label{sec:conclusion}

In this paper, we presented a complete specification of push-stream, a producer-driven object protocol for implementing concurrent streaming pipelines. Compared to the original design, our specification fixes two major issues, one preventing the ports of intermediate modules from terminating independently and one potentially leading to incorrect resuming of modules that are asynchronously pending. It also fully explicits the expected dynamic behaviour of modules, while the original readme description focused mostly on the static structure of modules.

In contrast to common programming libraries and frameworks based on a streaming model, the push-stream protocol relies solely on method calls and status flags on objects. It therefore does not require a scheduler, makes bounded buffers on communication ports optional, and does not require dynamic allocation on the heap for protocol interactions. It should be portable to lower-level languages and resource constrained environment.

We presented the specification of the protocol in multiple refinement steps, making each step and their associated properties easier to understand and verify than if it had been presented as a single artifact. We also derived by equivalence a definition of abstract modules that are closer to implementation code, to guide module implementers. From the abstract modules we then refined a module checker that automatically finds possible protocol violations in \tlaplus concrete module specifications. To our knowledge this is a first example of using equivalent specifications to support conflicting practical needs: i.e. concision on the one hand, and closeness to final implementation on the other.

Our results may inspire future work in different directions. The push-stream protocol could be extended to support modules with more than one input or output. A compiler could be written to automatically translate our concrete module specifications and module checker to common programming languages, thus removing the dependency on the \tlaplus toolbox for verification of module implementations. Networking protocols such as TCP and QUIC could be specified as sequence of refinement steps, to make their properties easier to understand and verify. The applicability of the protocol and possible necessary extensions to be compatible with preemptive multithreading and multi-threading in Python with the global interpreter lock deactivated~\cite{gross2023optionalgil} could also be studied.

We additionally provide the following supplementary material. Proofs of theorems are given in Appendix~\ref{appdx:proofs}. All \tlaplus specifications, including refinement steps and the module checker are given in Appendix~\ref{appdx:abstract-port-spec} to \ref{appdx:verifying-concrete-modules}. We also discuss how we ensured the correctness of specifications in Appendix~\ref{sec:correctness}, and compare the overhead of the protocol to common Python alternatives and Unix pipes in Appendix~\ref{sec:overhead}.

\acks
We would like to acknowledge the exceptional dedication and creativity of Dominic Tarr who designed the initial version of the push-stream protocol. The fact that it is practically useful  is the result of multiple years of refinement in the context of real-world deployment of peer-to-peer systems.

We thank Prof. Christian Tschudin for comments on an earlier version of this paper. We also want to specially acknowledge Prof. Christian Tschudin for providing significant freedom in choosing research problems, hiring us in a position which allowed us to teach extra-curricular seminars that contributed to deepen our understanding of the material, and finally allowing us to claim sole authorship of the paper. 

\newpage
\appendix

\section{Proofs}
\label{appdx:proofs}

\begin{theorem}
\label{thm:enc-implies-composability}
Encapsulation (Prop.~\ref{rq:encapsulation}) implies composability (Prop.~\ref{rq:composability}).
\end{theorem}
\begin{proof}
Suppose by contradiction, that encapsulation does not imply composability. This is equivalent to saying that two modules that individually implement the protocol correctly (including encapsulation), do not implement it correctly when connected. This implies that some interaction between the two modules triggers a protocol violation on any of the port(s) of the two connected modules. However, because of encapsulation, the only possible interactions between the two connected modules happen on the protocol port they share. Because both modules individually implement the protocol correctly, all interactions they generate on this port are valid. And since correct modules implement the protocol iff they only generate valid interactions on all their ports in reaction to valid interactions on any of their ports, there can be no protocol violation on other ports, which contradicts the assumption. Therefore, encapsulation implies composability.
\end{proof}

\newpage
\section{Abstract port}
\label{appdx:abstract-port-spec}

\tlapluspdfpage{spec/PushStreamPorts}{1}
\tlapluspdfpage{spec/PushStreamPorts}{2}

\newpage
\section{Input-output port}
\label{appdx:io-port-spec}

\tlapluspdfpage{spec/PushStreamInputOutputPorts2}{1}
\tlapluspdfpage{spec/PushStreamInputOutputPorts2}{2}
\tlapluspdfpage{spec/PushStreamInputOutputPorts2}{3}
\tlapluspdfpage{spec/PushStreamInputOutputPorts2}{4}
\tlapluspdfpage{spec/PushStreamInputOutputPorts2}{5}
\tlapluspdfpage{spec/PushStreamInputOutputPorts2}{6}
\tlapluspdfpage{spec/PushStreamInputOutputPorts2}{7}
\tlapluspdfpage{spec/PushStreamInputOutputPorts2}{8}
\tlapluspdfpage{spec/PushStreamInputOutputPorts2}{9}

\newpage
\section{Input-output port refines abstract port}
\label{appdx:refinement-io-port-to-abstract-port}

\tlapluspdfpage{spec/PushStreamIOP2RefinesPorts}{1}

\newpage
\section{Cross-port restrictions}
\label{appdx:cross-port-restrictions-spec}

\tlapluspdfpage{spec/PushStreamCrossPortRestrictions}{1}
\tlapluspdfpage{spec/PushStreamCrossPortRestrictions}{2}
\tlapluspdfpage{spec/PushStreamCrossPortRestrictions}{3}
\tlapluspdfpage{spec/PushStreamCrossPortRestrictions}{4}
\tlapluspdfpage{spec/PushStreamCrossPortRestrictions}{5}
\tlapluspdfpage{spec/PushStreamCrossPortRestrictions}{6}
\tlapluspdfpage{spec/PushStreamCrossPortRestrictions}{7}
\tlapluspdfpage{spec/PushStreamCrossPortRestrictions}{8}
\tlapluspdfpage{spec/PushStreamCrossPortRestrictions}{9}

\newpage
\section{Cross-port restrictions refine input-output port}
\label{appdx:refinement-cross-port-restrictions-to-io-port}

\tlapluspdfpage{spec/PushStreamCrossPortRestrictionsRefinesIOPort2}{1}

\newpage
\section{Abstract modules}
\label{appdx:abstract-modules-spec}

\tlapluspdfpage{spec/PushStreamAbstractModules}{1}
\tlapluspdfpage{spec/PushStreamAbstractModules}{2}
\tlapluspdfpage{spec/PushStreamAbstractModules}{3}
\tlapluspdfpage{spec/PushStreamAbstractModules}{4}
\tlapluspdfpage{spec/PushStreamAbstractModules}{5}
\tlapluspdfpage{spec/PushStreamAbstractModules}{6}
\tlapluspdfpage{spec/PushStreamAbstractModules}{7}
\tlapluspdfpage{spec/PushStreamAbstractModules}{8}
\tlapluspdfpage{spec/PushStreamAbstractModules}{9}
\tlapluspdfpage{spec/PushStreamAbstractModules}{10}
\tlapluspdfpage{spec/PushStreamAbstractModules}{11}
\tlapluspdfpage{spec/PushStreamAbstractModules}{12}
\tlapluspdfpage{spec/PushStreamAbstractModules}{13}
\tlapluspdfpage{spec/PushStreamAbstractModules}{14}
\tlapluspdfpage{spec/PushStreamAbstractModules}{15}
\tlapluspdfpage{spec/PushStreamAbstractModules}{16}
\tlapluspdfpage{spec/PushStreamAbstractModules}{17}
\tlapluspdfpage{spec/PushStreamAbstractModules}{18}
\tlapluspdfpage{spec/PushStreamAbstractModules}{19}
\tlapluspdfpage{spec/PushStreamAbstractModules}{20}

\newpage
\section{Equivalence between abstract modules and cross-port restrictions}
\label{appdx:abstract-modules-equiv}

\subsection{Abstract modules refine cross-port restrictions}
\label{appdx:abstract-modules-ref-crossports}
\tlapluspdfpage{spec/PushStreamAbstractModulesRefinesCrossPortRestrictions}{1}

\subsection{Cross-port restrictions refine abstract modules}
\label{appdx:crossports-ref-abstract-modules}
\tlapluspdfpage{spec/PushStreamCrossPortRestrictionsRefinesAbstractModules}{1}

\newpage
\section{Verifying concrete modules with abstract modules}
\label{appdx:verifying-concrete-modules}

\tlapluspdfpage{spec/PushStreamModuleChecker}{1}
\tlapluspdfpage{spec/PushStreamModuleChecker}{2}


\newpage
\section{Correctness of Specifications}
\label{sec:correctness}

A variety of detailed properties are listed at the end of all specifications in the previous Appendices. Some are equivalent to the properties listed in the protocol overview (Sec.~\ref{sec:protocol-overview}) and others are most specific. All were verified by the TLC model checker. 

A refinement is formalized by defining a mapping from every possible state of the input-output port specification to those of the abstract port. The refinements of Figure~\ref{fig:specification-overview}, shown either as implications or equivalence (mutual implication), are formalized in Appendices~\ref{appdx:refinement-io-port-to-abstract-port},~\ref{appdx:refinement-cross-port-restrictions-to-io-port}, and~\ref{appdx:abstract-modules-equiv}. The TLC model checker verifies that a refinement is valid by checking that the behaviour specification of the abstract port holds as a temporal property of all possible behaviours of the input-output port, under the refinement mapping (see video course~\cite{lamport2021tlaplus-video-course} and paper~\cite{lamport2017auxiliaryvariablestla}). 

For example, because of the refinement mapping, the input-output specification also meets the requirements met and the properties achieved by the abstract port specification.  By inspection, it is easy to verify, e.g., that invocation of the \texttt{write} method never occurs outside of the \textit{Writing} state, which prevents writing during pausing or termination; also, once terminating, it is not possible for the port to become paused, pausing or writing again (Prop~\ref{rq:irreversibility}). 

In the specification for cross-port restrictions (Appendix~\ref{appdx:cross-port-restrictions-spec}), all restricted actions are annotated with the property that contribute each restriction(s). Curious readers may re-trace some of our design steps by interactively querying the specification to provide example behaviours that are removed by the restrictions. They may do so by commenting some restriction conditions in the actions while asking TLC to verify the safety properties. TLC should report a temporal property violation with the name of the corresponding safety property, as well as an example behaviour that violates it. 

%

\newpage
\section{Overhead}
\label{sec:overhead}

Python offers many possibilities for processing a collection of values. This section compares them to pipelines of push-stream modules, both on their capabilities and performance overhead. We present the alternatives we considered, the methodology we used to compare them, our empirical results, and a discussion.

\subsection{Alternatives Compared}
\label{sec:alternatives}

The Python iterator protocol~\cite{python3iteratortype,yee2001python-iterator} is a reader-driven streaming protocol with only two methods: \texttt{\_\_iter\_\_()} that returns an object that implements the protocol (itself for an iterator object); and \texttt{\_\_next\_\_()} that returns the next value to process if available, or raises a \texttt{StopIteration} exception if not. Efficient processing pipelines can be created by composing functions from the \textit{itertools} library~\cite{python3itertools}: this will instantiate and connect iterator objects actually written in C. This approach is among the fastest available but lacks \textit{flow-control} (Prop.~\ref{rq:flow-control}) and \textit{graceful termination} (Prop.~\ref{rq:gracefulness}), i.e. propagation of termination to adjacent modules of the pipeline. We include it because it provides a useful performance target.

Python generators~\cite{schemenauer2001python-generators} are functions, which when called, return an iterator. They are distinguished from regular functions by the use the \texttt{yield} keyword to return values. Generators abstract the bookkeeping required to maintain state between requests for values, making iterators (coroutines) easier to write. An asynchronous version can be obtained by prefixing the function definition and for loop iterations in the caller with the \texttt{async} keyword~\cite{selivanov2015python-coroutine}. Generators can be faster than iterators written explicitly in Python but suffer the same limitation on flow-control and graceful termination because they are limited by the iterator protocol. Moreover, any generator that invokes an asynchronous generator, must either itself be asynchronous or must do extra bookkeeping.

Python supports invocation of subprocesses~\cite{python-subprocess}, both synchronously~\cite{python-sync-subprocess} and asynchronously~\cite{python-async-subprocess}, and may connect them with pipes on Unix. Pipeline modules may therefore be written in any language, and the pipeline orchestrated from Python, side-stepping most potential performance bottleneck that Python may introduce. However, synchronous subprocess creation precludes incrementally writing to the standard input of a process while incrementally reading from the standard output: the documentation explicitly mentions this may lead to deadlocks~\cite{python-sync-subprocess-deadlock}. Asynchronous subprocesses do not have this limitation and actually implement a subprocess protocol that mirrors the push-stream protocol~\cite{python-async-subprocess-protocol}, they therefore provide equivalent capabilities.\footnote{Asynchronous streaming protocols are some of the few concurrent building blocks in Python that have an explicit specification of their behaviour, as a state machine~\cite{python-async-streaming-protocol}. In our experience, they provide more reliable foundations for designing concurrent abstractions than many higher-level constructions.}

Reactive Python~\cite{reactive-python} is a popular library for creating producer-driven asynchronous streaming pipelines, with similar functional operators as itertools. It can also interoperate with the asyncio event loop, and therefore use asynchronous subprocesses or asynchronous generators as event generators. However, it does not support flow-control (which the author calls backpressure~\cite{reactive-python-doc-backpressure}), and orchestrates all operations with an event loop. Its synchronization protocol is therefore less expressive than push-stream but its performance profile provides an interesting reference point to what others have achieved purely in Python. We compare to the latest release, 4.1.0.

\subsection{Methodology}
We compare all alternatives on a simple processing pipeline, chosen to emphasize the performance cost of the protocol and runtime environment, e.g. even loop, garbage collector, etc.. Specifically, the modules do very little so that most of the execution time is spent in synchronization and activation of the different modules. For each alternative compared, we create a pipeline composed of the equivalent of a \texttt{Values} producer that outputs a sequence of increasing integers connected to a \texttt{Count} consumer (Table~\ref{tb:concrete-modules}). We also compare to a pipeline with an additional \texttt{Map} or \texttt{AsyncMap} transformer to measure the marginal cost of adding a transformer module that performs synchronous or asynchronous processing. 

We first measure the initialization time as libraries and the runtime system appear to lazily initialize, which makes the overhead for the first value orders of magnitude larger than for the next. We measure this time by executing the entire pipeline with a single value and report on the time needed to both create the pipeline, process the value, and terminate.

We also measure the time required for the entire pipeline to process increasingly larger collections of values by orders of magnitude, i.e. 1, 10, 100, etc. We report the normalized time, i.e. average time per value processed, at a scale where it stabilizes and the effect of initialization becomes negligible. We repeat the measurements ten times and report the average over the ten runs.

All our measurements were made on a Macbook Pro 13-inch 2019 with a 2.8GHz Intel Core i7 processor, with 16 GB of 2133MHz of LPDDR3 of memory, running MacOS Sequoia 15.7.4 and Python 3.12.9.

\subsection{Results}

We present our results in three groups: synchronous and Unix pipelines, which are the fastest; asynchronous which are slower; and various configurations of push-stream pipelines.

Figure~\ref{fig:results-sync} compares push-stream modules using only synchronous writing to iterators and generators, as well as to a Unix pipeline running as two subprocesses connected by a pipe and orchestrated by Python.

\begin{figure}[htbp]
\includegraphics[width=1.0\textwidth,page=1,trim=6cm 63cm 30cm 4cm, clip=true]{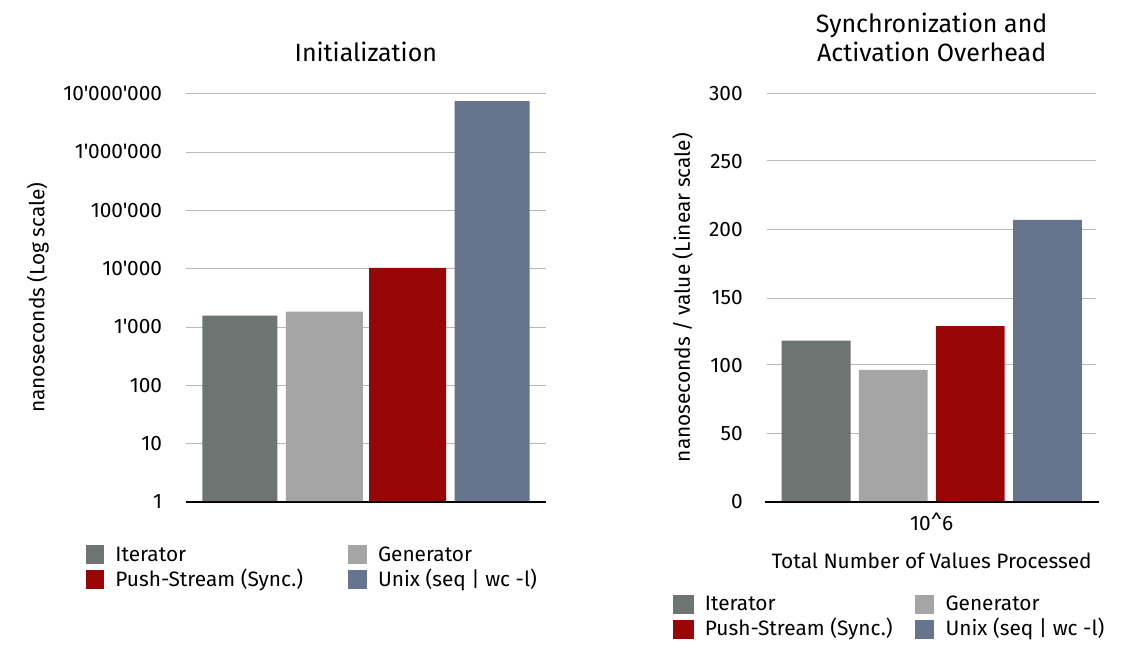}
\caption{\label{fig:results-sync} Comparison of push-stream's synchronous overhead to Python's built-in iterator and generator, and an equivalent Unix pipeline with a single pipe. Lower is better.}
\end{figure}

 The initialization time for iterators and generators is similar and the fastest, at one microsecond. The initialization of a push-stream pipeline is an order of magnitude slower, at ten microseconds. The initialization of a Unix pipeline is almost three orders of magnitude slower than for a push-stream at ten milliseconds.
 
 The synchronization and activation overhead stabilizes around $10^6$ values processed. The overhead of the push-stream protocol is similar to that of the iterator protocol, and roughly 25\% slower than iterators obtained with generators. The overhead of synchronizing through pipes for Unix pipelines is about 40\% higher.

Figure~\ref{fig:results-async} compares a push stream pipeline with an additional AsyncMap module (Table~\ref{tb:concrete-modules}) to equivalent pipelines written with Python async generators and Reactive Python. The push stream asynchronous map module is implemented with a long-running asynchronous task that reads a new future, i.e. single assignment variable object with synchronization operations~\cite{asyncio-rebooted}, for every new value.\footnote{Our first version of the push stream async map instead created a new task for every value and was 7.7 times slower.} The initialization time for push-stream is approximately in the middle in-between the built-in async generator and a reactive Python pipeline. The synchronization overhead for push-stream is a 2.5 times higher than the async generator but is more than 10 times lower than that of reactive Python.

\begin{figure}[htbp]
\includegraphics[width=1.0\textwidth,page=1,trim=6cm 51.5cm 30cm 15.5cm, clip=true]{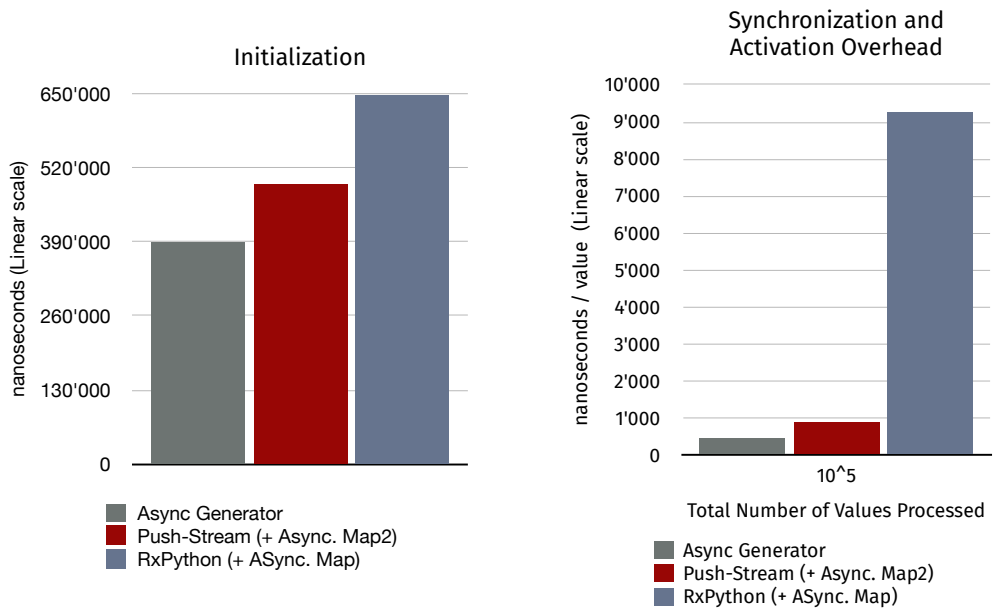}
\caption{\label{fig:results-async} Comparison of push-stream's asynchronous overhead to Python's built-in asynchronous generator and RxPython with a single map operator. Lower is better.}
\end{figure}

Figure~\ref{fig:results-push-stream-variations} shows how the overhead increases as different features of the push-stream protocol are used. Adding a synchronous map module triples the initialization time (scale is logarithmic) and multiplies the synchronization overhead by 1.5. Adding an asynchronous map module increases the initialization time by almost two orders of magnitude and multiplies the synchronization overhead by 8. Creating a push stream pipeline of subprocess modules communicating through the push-stream protocol, with a module in-between that parses and splits standard output lines before passing them to the next subprocess, increases the initialization time by three orders of magnitude and the synchronization overhead by 8. Compared to the previous Unix pipeline that connects both processes with a pipe (Figure~\ref{fig:results-sync}), the initialization time is similar but the synchronization overhead is 5 times larger.

\begin{figure}[htbp]
\includegraphics[width=1.0\textwidth,page=1,trim=5.5cm 37.5cm 30cm 29cm, clip=true]{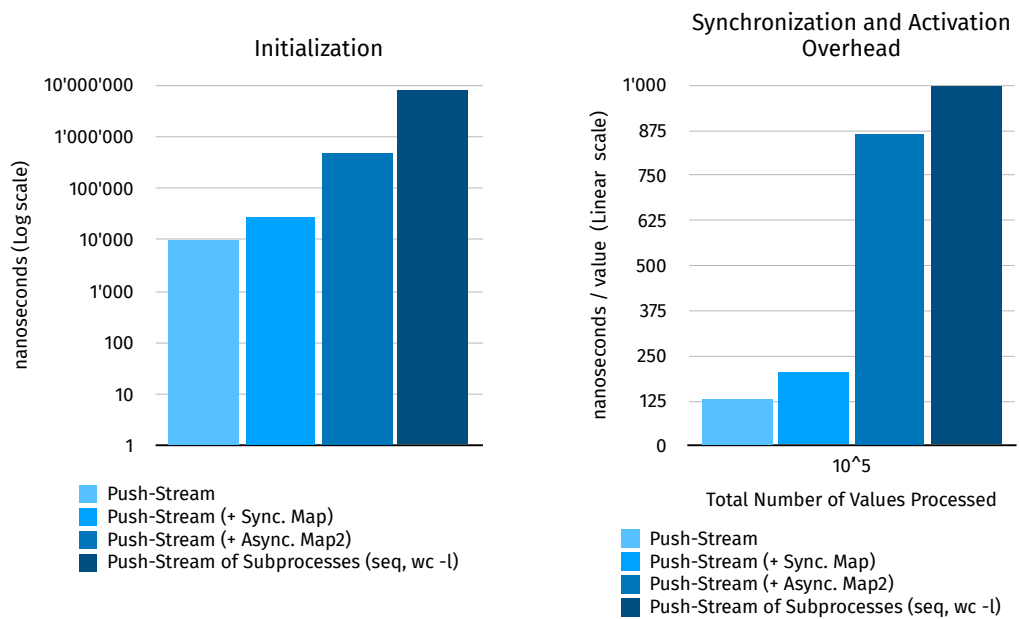}
\caption{\label{fig:results-push-stream-variations} Marginal cost of adding an additional synchronous module, asynchronous module, or subprocess delegation in push-stream pipelines.  Lower is better.}
\end{figure}

\subsection{Discussion}
In synchronous mode, a push-stream pipeline has similar overhead as iterators objects and generators, but with almost an order of magnitude larger initialization time. The difference in initialization comes from the construction of the pipeline. The synchronization overhead is similar because it boils down to the invocation of the \texttt{write} method. It surprises us that the method call overhead in Python used to implement the protocol turns out lower than the synchronization overhead for Unix processes to communicate over pipes, i.e. bounded buffers~\cite{dijkstra123coop-seq-procs} that behave as files. This is explained by the fact that a method call within the same process has less overhead than inter-process synchronization through pipes. In combination with the initialization time that is three orders of magnitude smaller, this means that processing pipelines with low amount of computation per value and a medium number of values to process (less than $10^4$) may be faster to execute directly in Python than by spanning a Unix pipeline.

In asynchronous mode, a push-stream pipeline has 2.5 times the overhead of a chain of async generators but this appears to be the cost for the ability to intermix easily with asynchronous and synchronous modules, which adds the overhead of at least the \texttt{write} method call in addition to performing the asynchronous protocol from outside an asynchronous context. Nonetheless, this extra overhead is an order of magnitude less than for reactive python. That suggests reactive Python could add flow-control and lower the overhead of reactive python pipelines significantly by using our protocol.

Our results are obviously unrepresentative of actual workloads because the pipelines are not doing any useful work. The practicality of each alternative should be assessed by comparing the protocols' overhead to the application-specific processing time per value. This may reveal that for certain applications with high processing time per value, in the end, the performance difference between alternatives may not matter much. For applications that do, our push-stream protocol provides additional flow-control and graceful termination at an overhead within a factor of 1-3 of built-in protocols.

\newpage
\printbibliography

\end{document}